\documentclass[twoside,english,twocolumn,5p]{elsarticle}
\usepackage[T1]{fontenc}
\usepackage[latin9]{inputenc}
\usepackage{xcolor}
\usepackage{amsmath}
\usepackage{amsthm}
\usepackage{amssymb}
\usepackage{graphicx}
\usepackage{esint}

\makeatletter
\theoremstyle{remark}
\newtheorem{rem}{\protect\remarkname}
\theoremstyle{definition}
\newtheorem{defn}{\protect\definitionname}
\theoremstyle{plain}
\newtheorem{lem}{\protect\lemmaname}
\ifx\proof\undefined\
  \newenvironment{proof}[1][\proofname]{\par
    \normalfont\topsep6\p@\@plus6\p@\relax
    \trivlist
    \itemindent\parindent
    \item[\hskip\labelsep
          \scshape
      #1]\ignorespaces
  }{%
    \endtrivlist\@endpefalse
  }
  \providecommand{\proofname}{Proof}
\fi
\theoremstyle{plain}
\newtheorem{thm}{\protect\theoremname}
\theoremstyle{plain}
\newtheorem{cor}{\protect\corollaryname}

\journal{Automatica}

\usepackage[caption=false,font=footnotesize]{subfig}

\usepackage{amstext}
\usepackage{amsthm}

\usepackage{epstopdf}

\usepackage{tikz}
\usetikzlibrary{arrows}
\usetikzlibrary{shadows}
\tikzset{
  every overlay node/.style={
    draw=white,anchor=north west,
  },
}
\usetikzlibrary{positioning, decorations.pathmorphing}

\usepackage{hyperref}
\@ifundefined{rangeHsb}{\usepackage{xcolor}}{}

\usepackage{algorithm,algpseudocode}

\date{May 18, 2021}

\usepackage{lipsum}

\usepackage{scrlayer-scrpage}
\clearpairofpagestyles
\makeatother

\usepackage{babel}
\providecommand{\corollaryname}{Corollary}
\providecommand{\definitionname}{Definition}
\providecommand{\lemmaname}{Lemma}
\providecommand{\remarkname}{Remark}
\providecommand{\theoremname}{Theorem}

\begin{document}

\begin{frontmatter}{}

\title{Cluster Consensus on Matrix-weighted Switching Networks}

\author[rvt]{Lulu~Pan}

\author[rvt]{Haibin~Shao \corref{cor1}}

\author[focal]{Mehran~Mesbahi}

\author[rvt]{Dewei~Li}

\author[rvt]{Yugeng~Xi}

\cortext[cor1]{Corresponding author.}

\address[rvt]{Department of Automation, Shanghai Jiao Tong University and Key Laboratory
of System Control and Information Processing, \\Ministry of Education,
Shanghai, China 200240}

\address[focal]{William E. Boeing Department of Aeronautics and Astronautics, University
of Washington,  Seattle, WA, USA 98195-2400}
\begin{abstract}
This paper examines the cluster consensus problem of multi-agent systems
on matrix-weighted switching networks. Necessary and/or sufficient
conditions under which cluster consensus can be achieved are obtained
and quantitative characterization of the steady-state of the cluster
consensus are provided as well. Specifically, \textcolor{black}{if
the underlying network switches amongst finite number of networks,
a necessary condition for cluster consensus of multi-agent system
on switching matrix-weighted networks is firstly presented, it is
shown that the steady-state of the system lies in the intersection
of the null space of matrix-valued Laplacians corresponding to all
switching networks. Second, if the underlying network switches amongst
infinite number of networks, the matrix-weighted integral network
is employed to provide sufficient conditions for cluster consensus
and the quantitative characterization of the corresponding steady-state
of the multi-agent system, using null space analysis of matrix-valued
Laplacian related of integral network associated wit}h the switching
networks. In particular, conditions for the bipartite consensus under
the matrix-weighted switching networks are examined. Simulation results
are finally provided to demonstrate the theoretical analysis.
\end{abstract}

\end{frontmatter}{}

\section{Introduction}

Achieving consensus via local interactions amongst agents turns out
to be an important paradigm in distributed control of multi-agent
networks \citet{mesbahi2010graph,olfati2007consensus,degroot1974reaching,zhang2018fully}.
However, it has long been assumed that the edges, representing the
interaction between neighboring agents, are weighted by scalars, which
apparently ignores the interdependencies among multi-dimensional states
of neighboring agents. Recently, matrix-weighted network is introduced
to characterize the complicated interactions amongst the high-dimensional
states of agents \citet{trinh2018matrix,sun2018dimensional,pan2018bipartite,tuna2016synchronization}.
In fact, matrix-weighted networks naturally arise in scenarios such
as graph effective resistance and its applications in distributed
control and estimation \citet{tuna2017observability,barooah2008estimation},
opinion dynamics \citet{friedkin2016network,ye2020continuous}, bearing-based
formation control \citet{zhao2015translational}, coupled oscillators
dynamics \citet{tuna2019synchronization}, and consensus and synchronization
\citet{trinh2018matrix,tuna2016synchronization,pan2018bipartite,su2019bipartite}. 

As opposed to scalar-weighted networks, network connectivity does
not translate to achieving consensus for matrix-weighted networks.
Rather than achieving consensus, the expansion of the null space of
the associated matrix-valued Laplacian enables the multi-agent system
on matrix-weighted networks to achieve cluster consensus (or clustering),
even if the underlying network is connected \citet{trinh2018matrix}.
This elegant property enables one to design desired cluster structures
of for multi-agent systems by elaborately investigating the connection
between the matrix-valued edge weights and the null space of matrix-valued
Laplacian. Notably, achieving the cluster synchronization in coupled
oscillator systems has been shown to be closely related to memory
process in human brain \citet{skardal2011cluster,ashwin2007dynamics,fell2011role,hipp2012large}.
However, achieving desired cluster consensus is not trivial in the
case of scalar-weighted networks, where network-wide information and
specific control strategies have to be involved \citet{wu2008cluster,xia2011clustering}.
In contrast, cluster consensus can be naturally achieved under matrix-weighted
networks. Recently, the conditions of achieving cluster consensus
on fixed matrix-weighted networks were reported in \citet{trinh2018matrix}.
\textcolor{black}{Nevertheless, the underlying network of a multi-agent
system can vary over time in a great variety of situations \citet{cao2008reaching,olfati2004consensus,ren2005consensus,cao2011necessary,meng2018uniform,anderson2016convergence,cao2012overview}.
}The continuous-time and discrete-time consensus problem on time-varying
matrix-weighted networks are discussed \citet{Pan2021Tac,van2020discrete}.
Nevertheless, to the best of our knowledge, the conditions for cluster
consensus on dynamic matrix-weighted networks are still lacking. 

This paper is intended to provide quantitative characterization for
cluster consensus of multi-agent systems on matrix-weighted switching
networks. \textcolor{black}{The contribution of this paper is threefold.
First, for the case that the underlying network switches amongst finite
number of networks, necessary condition for the cluster consensus
of matrix-weighted switching networks has been exploited and an essential
connection between cluster consensus value and the Laplacian matrices
of matrix weighted switching networks has been established. Second,
for the case that the underlying network switches amongst infinite
number of networks, sufficient conditions for cluster consensus are
obtained by examining the structure of null spaces associated to the
matrix-valued Laplacian of the associated integral network. Finally,
we provide conditions for a class of specific cluster consensus, namely
bipartite consensus, under the matrix-weighted switching networks,
and graph-theoretic condition is obtained. The }results obtained in
this paper provide further insight into the collective behavior of
multi-agent systems.

\textcolor{black}{The remainder of the paper is organized as follows.
Preliminaries and problem formulation are introduced in $\mathsection$\ref{sec:Preliminaries}
and $\mathsection$\ref{sec:Problem-Formulation}, respectively, followed
by the cluster consensus and bipartite consensus conditions in $\mathsection$\ref{sec:Consensus-on-General}.
Simulation examples are presented in $\mathsection$\ref{sec:Simulation-Results};
we provide concluding remarks in $\mathsection$\ref{sec:Conclusion}.}

\section{\textcolor{black}{Preliminaries\label{sec:Preliminaries}}}

\subsection{\textcolor{black}{Notations}}

\textcolor{black}{Let $\mathbb{R}$, $\mathbb{N}$ and $\mathbb{Z}_{+}$
be the set of real numbers, natural numbers and positive integers,
respectively. For $n\in\mathbb{Z}_{+}$, denote $\underline{n}=\left\{ 1,2,\ldots,n\right\} $.
A symmetric matrix $M\in\mathbb{R}^{n\times n}$ is positive definite(Resp.
negative definite), denoted by $M\succ0$ (Resp. $M\prec0$), if $\boldsymbol{z}^{\top}M\boldsymbol{z}>0$
(Resp. $\boldsymbol{z}^{\top}M\boldsymbol{z}<0$) for all nonzero
$\boldsymbol{z}\in\mathbb{\mathbb{R}}^{n}$, and is positive (Resp.
negative) semi-definite, denoted by $M\succeq0$ (Resp. $M\preceq0$),
if $\boldsymbol{z}^{\top}M\boldsymbol{z}\ge0$ (Resp. $\boldsymbol{z}^{\top}M\boldsymbol{z}\le0$)
for all $\boldsymbol{z}\in\mathbb{\mathbb{R}}^{n}$. The null space
of a matrix $M\in\mathbb{R}^{n\times n}$ is denoted by $\text{{\bf null}}(M)=\left\{ \boldsymbol{z}\in\mathbb{R}^{n}|M\boldsymbol{z}=\boldsymbol{0}\right\} $.
$\boldsymbol{1_{n}}\in\mathbb{R}^{n}$ and $0_{d\times d}\in\mathbb{R}^{d\times d}$
designate the vector whose components are all $1$'s and the matrix
whose components are all $0$'s, r}espectively. Let $\left\lfloor x\right\rfloor $
denote the greatest integer less than or equal to $x\in\mathbb{R}$.
The sign function $\text{{\bf sgn}}(\cdot):\mathbb{R}^{n\times n}\mapsto\left\{ 0,-1,1\right\} $
satisfies $\text{{\bf sgn}}(M)=1$ if $M\succeq0$ or $M\succ0$,
$\text{{\bf sgn}}(M)=-1$ if $M\preceq0$ or $M\prec0$, and $\text{{\bf sgn}}(M)=0$
if $M=0_{d\times d}$. 

\subsection{Graph Theory}

A matrix-weighted switching graph is denoted by $\mathcal{G}(t)=(\mathcal{V},\mathcal{E}(t),A(t))$,
where $t$ refers to time index. The node and edge sets of $\mathcal{G}$
are denoted by $\mathcal{V}=\left\{ 1,2,\ldots,n\right\} $ and $\mathcal{E}(t)\subseteq\mathcal{V}\times\mathcal{V}$,
respectively. The weight on the edge $(i,j)\in\mathcal{E}(t)$ is
encoded by the symmetric matrix $A_{ij}(t)\in\mathbb{R}^{d\times d}$
such that $\mid A_{ij}(t)\mid\succeq0$ or $\mid A_{ij}(t)\mid\succ0$,
and $A_{ij}(t)=0_{d\times d}$ for $(i,j)\not\in\mathcal{E}(t)$.\textbf{
}Thereby, the matrix-valued adjacency matrix $A(t)=[A_{ij}(t)]\in\mathbb{R}^{dn\times dn}$
is a block matrix such that the block located in its $i$-th row and
the $j$-th column is $A_{ij}(t)$. It is assumed that $A_{ij}(t)=A_{ji}(t)$
for all $i\not\not=j\in\mathcal{V}$ and $A_{ii}(t)=0_{d\times d}$
for all $i\in\mathcal{V}$. A bipartition of node set $\mathcal{V}$
of matrix-weighted switching graph $\mathcal{G}(t)=(\mathcal{V},\mathcal{E}(t),A(t))$
at time $t$ is two subsets of nodes $\mathcal{V}_{i}\subset\mathcal{V}$
where $i\in\underline{2}$ such that $\mathcal{V}_{1}\cup\mathcal{V}_{2}=\mathcal{V}$
and $\mathcal{V}_{1}\cap\mathcal{V}_{2}=\textrm{Ø}$. 

A path of $\mathcal{G}(t)$ at time $t$ is a sequence of edges of
the form $(i_{1},i_{2}),(i_{2},i_{3}),\ldots,(i_{p-1},i_{p})$, where
nodes $i_{1},i_{2},\ldots,i_{p}\in\mathcal{V}$ are distinct; in this
case we say that node $i_{p}$ is reachable from $i_{1}$. The graph
$\mathcal{G}(t)$ is connected at time $t$ if any two distinct nodes
in $\mathcal{G}(t)$ are reachable from each other. A tree is a connected
graph with $n\ge2$ nodes a\textcolor{black}{nd $n-1$ edges where
$n\in\mathbb{Z}_{+}$. For matrix-weighted }switching\textcolor{black}{{}
networks, we adopt the following terminology. An edge $(i,j)\in\mathcal{E}(t)$
at time $t$ is positive (respectively, negative) definite or positive
(respectively, negative) semi-definite if the corresponding edge weight
$A_{ij}(t)$ is positive (respectively, negative) definite or positive
(respectively, negative) semi-definite. A positive-negative path of
$\mathcal{G}(t)$ at time $t$ is a path such that every edge in this
path is either positive definite or negative definite. A positive-negative
tree of $\mathcal{G}(t)$ at time $t$ is a tree such that every edge
in this tree is either positive definite or negative definite. A positive-negative
spanning tree of $\mathcal{G}(t)$ at time $t$ is a positive-negative
tree containing all nodes in $\mathcal{G}(t)$. }

\section{\textcolor{black}{Problem Formulation\label{sec:Problem-Formulation}}}

\textcolor{black}{Consider a multi-agent system consisting of $n>1$
($n\in\mathbb{Z}_{+}$) agents whose interaction network is characterized
by a matrix-weighted switching graph $\mathcal{G}(t)=(\mathcal{V},\mathcal{E}(t),A(t))$.
Denote the state of an agent $i\in\mathcal{V}$ as $\boldsymbol{x}_{i}(t)=[x_{i1}(t),\ldots,x_{id}(t)]^{\top}\in\mathbb{R}^{d}$,
evolving according to the protocol 
\begin{equation}
\dot{\boldsymbol{x}}_{i}(t)=-\sum_{j\in\mathcal{N}_{i}(t)}\mid A_{ij}(t)\mid(\boldsymbol{x}_{i}(t)-\text{{\bf sgn}}(A_{ij}(t))\boldsymbol{x}_{j}(t)),\thinspace i\in\mathcal{V},\label{equ:matrix-consensus-protocol}
\end{equation}
where $\mathcal{N}_{i}(t)=\left\{ j\in\mathcal{V}\,|\,(i,j)\in\mathcal{E}(t)\right\} $
denotes the neighbor set of agent $i\in\mathcal{V}$ at time $t$.}
\textcolor{black}{Note that \eqref{equ:matrix-consensus-protocol}
degenerates into the scalar-weighted case when $A_{ij}(t)=a_{ij}(t)I$,
where $a_{ij}(t)\in\mathbb{R}$ and $I$ denotes the $d\times d$
identity matrix. }

\textcolor{black}{Let $D(t)=\text{{\bf diag}}\left\{ D_{1}(t),\cdots,D_{n}(t)\right\} \in\mathbb{R}^{dn\times dn}$
be the matrix-valued degree matrix of $\mathcal{G}(t)$, where $D_{i}(t)=\sum_{j\in\mathcal{N}_{i}}\mid A_{ij}(t)\mid\in\mathbb{R}^{d\times d}$
and $i\in\mathcal{V}$. The matrix-valued Laplacian is subsequently
defined as $L(t)=D(t)-A(t)$. The dynamics of the overall multi-agent
system now admits the form,
\begin{equation}
\dot{\boldsymbol{x}}(t)=-L(t)\boldsymbol{x}(t),\label{equ:matrix-consensus-overall}
\end{equation}
where $\boldsymbol{x}(t)=[\boldsymbol{x}_{1}^{\top}(t),\ldots,\boldsymbol{x}_{n}^{\top}(t)]^{\top}\in\mathbb{R}^{dn}$.}
\begin{rem}
It is well-known that network connectivity plays a central role in
determining consensus for scalar-weighted time-invariant networks
\citet{olfati2004consensus}. \textcolor{black}{However, for the matrix-weighted
time-invariant networks,} network connectivity is only a necessary
condition for consensus on matrix-weighted networks, it is possible
to achieve cluster consensus even if the network is connected, which
is related with the properties of the null space of matrix-valued
Laplacian matrix for time-invariant network \citet{trinh2018formation}.
\end{rem}
\begin{defn}[Cluster Consensus]
\textcolor{black}{\label{def:cluster-consensus} The multi-agent
system \eqref{equ:matrix-consensus-overall} admits cluster consensus,
if there exists a partition of node set $\mathcal{V}$, say $\mathcal{V}_{1},\ldots,\mathcal{V}_{l}$
where $l\in\mathbb{Z}_{+}$ and $l\le n$, such that all agents belonging
to the same partition achieve consensus, while for any two agents
$i$ and $j$ belonging to two different partitions, ${\color{black}{\color{black}{\color{blue}{\color{red}{\color{black}\lim{}_{t\rightarrow\infty}\boldsymbol{x}_{i}(t)}}}\neq{\color{blue}{\color{red}{\color{black}\lim{}_{t\rightarrow\infty}\boldsymbol{x}_{j}(t)}}}}}$.
Each $\mathcal{V}_{i},\,i\in\underline{l}$ is referred to as a cluster.
In particular, the cluster consensus is referred to as consensus and
bipartite consensus if $l=1$ and $l=2$, respectively.}
\end{defn}
\textcolor{black}{The }\textbf{\textcolor{black}{Definition}}\textcolor{black}{{}
\ref{def:cluster-consensus} implies that there exists a vector $\boldsymbol{x}^{*}\in\mathbb{R}^{dn}$
such that $\lim_{t\rightarrow\infty}\boldsymbol{x}(t)=\boldsymbol{x}^{*}$,
where $\boldsymbol{x}^{*}$ is influenced by the initial states of
the multi-agent system \eqref{equ:matrix-consensus-overall}. This
work aims to investigate conditions under which the cluster consensus
state of mul}ti-agent system \eqref{equ:matrix-consensus-overall}
on matrix-weighted switching networks can be quantitatively characterized.
We\textcolor{black}{{} adopt the following assumptions on the underlying
matrix-weighted switching network \citet{olfati2004consensus,ren2005consensus,cao2011necessary}.}

\textbf{\textcolor{black}{Assumption 1.}}\textcolor{black}{{} There
exists a sequence $\{t_{k}|k\in\mathbb{N}\}$ such that $\lim_{k\rightarrow\infty}t_{k}=\infty$
and the dwell time satisfies $\triangle t_{k}=t_{k+1}-t_{k}\geq\alpha$
for all $k\in\mathbb{N}$, where $\alpha>0$, $t_{0}=0$, and $\mathcal{G}(t)$
is time-invariant for $t\in[t_{k},t_{k+1})$ for all $k\in\mathbb{N}$. }

\textbf{\textcolor{black}{Assumption 2.}}\textcolor{black}{{} In addition
to }\textbf{\textcolor{black}{Assumption }}\textcolor{black}{1, the
switching networks $\mathcal{G}(t)$ is chosen from a finite set $\left\{ \mathcal{G}_{1},\mathcal{G}_{2},\ldots,\mathcal{G}_{M}\right\} $
for some $M\in\mathbb{Z}_{+}$, and $\mathcal{G}_{i},\,i\in\underline{M}$
appears in the sequence of $\mathcal{G}(t)$ for infinitely times.}

\section{\textcolor{black}{Main Results \label{sec:Consensus-on-General}}}

In the following part, under the condition that the cluster consensus
can be achieved for the multi-agent system \eqref{equ:matrix-consensus-overall},
we shall exploit the connection between the cluster consensus and
the null space of matrix-valued Laplacian matrices associated with
a sequence of matrix-weighted networks.\textcolor{black}{{} We shall
start from the case that the underlying network of multi-agent system
\eqref{equ:matrix-consensus-overall} switches amongst finite number
of networks, as stated in the Assumption 2. Before showing the main
result of this part, we first explore the properties of $\underset{t\rightarrow\infty}{\text{{\bf lim}}}\boldsymbol{x}(t)$,
which plays an important role in the proof of our main result.}
\begin{lem}
\textcolor{black}{Consider the multi-agent system \eqref{equ:matrix-consensus-overall}
on a matrix-weighted switching network $\mathcal{G}(t)$ satisfying
}\textbf{\textcolor{black}{Assumption }}\textcolor{black}{2. If the
multi-agent system \eqref{equ:matrix-consensus-overall} achieves
the cluster consensus, namely, there exists $\boldsymbol{x}^{*}\in\mathbb{R}^{dn}$
such that $\underset{t\rightarrow\infty}{\text{{\bf lim}}}\boldsymbol{x}(t)=\boldsymbol{x}^{*}$,
then $\underset{t\rightarrow\infty}{\text{{\bf lim}}}L(t)\boldsymbol{x}^{*}=\boldsymbol{0}$.}
\end{lem}
\begin{proof}
Denote by $\varPhi(t,0)$ as the state transition matrix of multi-agent
system \eqref{equ:matrix-consensus-overall} over time interval $[0,t]$,
then one has $\parallel\varPhi(t,0)\parallel\leq1$, thus for any
$\boldsymbol{x}(0)\in\mathbb{R}^{dn}$ and $t>0$, \textcolor{black}{
\[
\parallel\boldsymbol{x}(t)\parallel\leq\parallel\boldsymbol{x}(0)\parallel.
\]
In addition, based on $\dot{\boldsymbol{x}}(t)=-L(t)\boldsymbol{x}(t)$,
one can derive $\ddot{\boldsymbol{x}}(t)=\left(L(t)\right)^{2}\boldsymbol{x}(t)$,
therefore, 
\[
\parallel\ddot{\boldsymbol{x}}_{i}(t)\parallel\leq\parallel\ddot{\boldsymbol{x}}(t)\parallel\leq\beta^{2}\parallel\boldsymbol{x}(t)\parallel\leq\beta^{2}\parallel\boldsymbol{x}(0)\parallel,
\]
where $\beta=\underset{\{i\in\underline{M}\}}{\text{{\bf max}}}\parallel L_{i}\parallel$
and $L_{i}$ is the matrix-valued Laplacian matrix corresponding to
$\mathcal{G}_{i}$, where $i\in\underline{M}$. Therefore, according
to }\textbf{\textcolor{black}{Lemma }}\textcolor{black}{\ref{lem:convergence lemma}}\textbf{\textcolor{black}{{}
}}\textcolor{black}{in the Appendix}\textbf{\textcolor{black}{, }}\textcolor{black}{one
has}\textbf{\textcolor{black}{,}}\textcolor{black}{
\[
\underset{t\rightarrow\infty}{\text{{\bf lim}}}\dot{\boldsymbol{x}}_{i}(t)=\boldsymbol{0},
\]
which imply that $\underset{t\rightarrow\infty}{\text{{\bf lim}}}\dot{\boldsymbol{x}}(t)=\boldsymbol{0}$.
Due to 
\[
L(t)\boldsymbol{x}^{*}=\left(\dot{\boldsymbol{x}}(t)+L(t)\boldsymbol{x}^{*}\right)-\dot{\boldsymbol{x}}(t),
\]
 and 
\begin{align*}
\parallel\dot{\boldsymbol{x}}(t)+L(t)\boldsymbol{x}^{*}\parallel & =\parallel-L(t)\left(\boldsymbol{x}(t)-\boldsymbol{x}^{*}\right)\parallel\\
 & \leq\beta\parallel\boldsymbol{x}(t)-\boldsymbol{x}^{*}\parallel,
\end{align*}
 thus, $\underset{t\rightarrow\infty}{\text{{\bf lim}}}L(t)\boldsymbol{x}^{*}=\boldsymbol{0}$.}
\end{proof}
\textcolor{black}{}
\textcolor{black}{Denote the state transition matrix of multi-agent
system \eqref{equ:matrix-consensus-overall} over time interval $[t_{0},t]$
as $\Phi(t,t_{0})$, then $\boldsymbol{x}(t)=\Phi(t,t_{0})\boldsymbol{x}(t_{0})$,
the following lemma presents the properties of $\Phi(t,t_{0})$ which
decides the convergence value of $\boldsymbol{x}(t)$.}
\begin{lem}
\textcolor{black}{Consider the multi-agent system \eqref{equ:matrix-consensus-overall}
on a matrix-weighted switching network $\mathcal{G}(t)$ satisfying
}\textbf{\textcolor{black}{Assumption }}\textcolor{black}{2. Then,
$\underset{t\rightarrow\infty}{\text{{\bf lim}}}\boldsymbol{x}(t)$
exists for any $\boldsymbol{x}(t_{0})\in\mathbb{R}^{dn}$ if and only
if $\underset{t\rightarrow\infty}{\text{{\bf lim}}}\Phi(t,t_{0})$
exists. Moreover, denote by $\underset{t\rightarrow\infty}{\text{{\bf lim}}}\Phi(t,t_{0})=\Phi^{*}(t_{0})$,
then $\left[\Phi^{*}(t_{0})\right]^{i}=\Phi^{*}(t_{0})$ for any $i\in\mathbb{Z}_{+}$.}
\end{lem}
\begin{proof}
\textcolor{black}{(Sufficiency) Due to $\boldsymbol{x}(t)=\Phi(t,t_{0})\boldsymbol{x}(t_{0})$
and $\underset{t\rightarrow\infty}{\text{{\bf lim}}}\Phi(t,t_{0})$
exists, thus, 
\[
\underset{t\rightarrow\infty}{\text{{\bf lim}}}\boldsymbol{x}(t)=\underset{t\rightarrow\infty}{\text{{\bf lim}}}\Phi(t,t_{0})\boldsymbol{x}(t_{0})=\Phi^{*}(t_{0})\boldsymbol{x}(t_{0}),
\]
i.e., $\underset{t\rightarrow\infty}{\text{{\bf lim}}}\boldsymbol{x}(t)$
exists.}

\textcolor{black}{(Necessity) If $\underset{t\rightarrow\infty}{\text{{\bf lim}}}\boldsymbol{x}(t)$
exists for any $\boldsymbol{x}(t_{0})\in\mathbb{R}^{dn}$, without
loss of generality, one can choose $\boldsymbol{x}(t_{0})=\boldsymbol{e}_{i},\,i\in\underline{dn}$,
where $\boldsymbol{e}_{i}\in\mathbb{R}^{dn}$ has its $i$-th component
equal to one with others equal to zero. Then, one has,
\begin{align*}
 & \underset{t\rightarrow\infty}{\text{{\bf lim}}}\Phi(t,t_{0})\\
= & \left[\underset{t\rightarrow\infty}{\text{{\bf lim}}}\Phi(t,t_{0})\boldsymbol{e}_{1},\ldots,\underset{t\rightarrow\infty}{\text{{\bf lim}}}\Phi(t,t_{0})\boldsymbol{e}_{dn}\right],
\end{align*}
due to $\underset{t\rightarrow\infty}{\text{{\bf lim}}}\boldsymbol{x}(t)$
exists for any $\boldsymbol{x}(t_{0})\in\mathbb{R}^{dn}$ and $\underset{t\rightarrow\infty}{\text{{\bf lim}}}\boldsymbol{x}(t)=\underset{t\rightarrow\infty}{\text{{\bf lim}}}\Phi(t,t_{0})\boldsymbol{x}(t_{0})$,
one can conclude that $\underset{t\rightarrow\infty}{\text{{\bf lim}}}\Phi(t,t_{0})$
exists.}

Denote by $\underset{t\rightarrow\infty}{\text{{\bf lim}}}\boldsymbol{x}(t)=\boldsymbol{x}^{*}$,
due to the fact $\underset{t\rightarrow\infty}{\text{{\bf lim}}}L(t)\boldsymbol{x}^{*}=\boldsymbol{0}$
and $\mathcal{G}_{i},\,i\in\underline{M}$ appears in the sequence
of $\mathcal{G}(t)$ for infinitely times, one has $L(t)\boldsymbol{x}^{*}=\boldsymbol{0}$,
therefore, for any $\boldsymbol{x}(t_{0})\in\mathbb{R}^{dn}$, 
\[
L(t)\boldsymbol{x}^{*}=L(t)\Phi^{*}(t_{0})\boldsymbol{x}(t_{0})=\boldsymbol{0},
\]
then, 
\begin{align*}
 & L(t)\Phi^{*}(t_{0})\\
= & \left[L(t)\Phi^{*}(t_{0})\boldsymbol{e}_{1},\ldots,L(t)\Phi^{*}(t_{0})\boldsymbol{e}_{dn}\right]\\
= & 0_{dn\times dn}.
\end{align*}

\textcolor{black}{According to the Peano-Baker series form of $\Phi(t,t_{0})$,
\begin{align*}
 & \Phi(t,t_{0})\\
= & I_{n}+\sum_{k=1}^{\infty}\int_{t_{0}}^{t}\left[-L(\sigma_{1})\right]\int_{t_{0}}^{\sigma_{1}}\left[-L(\sigma_{2})\right]\cdots\\
 & \int_{t_{0}}^{\sigma_{k-1}}\left[-L(\sigma_{k})\right]d\sigma_{k}\cdots d\sigma_{2}d\sigma_{1},
\end{align*}
thus, $\Phi(t,t_{0})\Phi^{*}(t_{0})=\Phi^{*}(t_{0})$. Take the limit
on both sides leads to $\left[\Phi^{*}(t_{0})\right]^{2}=\Phi^{*}(t_{0})$,
therefore,}\textcolor{orange}{{} ${\color{black}\left[\Phi^{*}(t_{0})\right]^{i}=\Phi^{*}(t_{0})}$
}\textcolor{black}{for any}\textcolor{orange}{{} }\textcolor{black}{$i\in\mathbb{Z}_{+}$.}
\end{proof}
\textcolor{black}{}
\textcolor{black}{Based on the above established Lemmas, we shall
show the relationship between the c}luster consensus and the matrix-valued
Laplacian matrix $L(t)$ of the associated matrix-weighted networks,
and further provide the explicit expression of the cluster consensus
value.
\begin{thm}
\textcolor{black}{\label{lem:cluster-value}Consider the multi-agent
system \eqref{equ:matrix-consensus-overall} on a matrix-weighted
switching network $\mathcal{G}(t)$ satisfying }\textbf{\textcolor{black}{Assumption
}}\textcolor{black}{2. If $\underset{t\rightarrow\infty}{\text{{\bf lim}}}\boldsymbol{x}(t)=\boldsymbol{x}^{*}$,
then 
\[
\boldsymbol{x}^{*}\in\underset{i\in\underline{M}}{\bigcap}\text{{\bf null}}(L(\mathcal{G}_{i})).
\]
Moreover,
\[
\boldsymbol{x}^{*}=\sum_{i=1}^{r}(\boldsymbol{\eta}_{i}^{\top}\boldsymbol{x}(t_{0}))\boldsymbol{\eta}_{i},
\]
where $\boldsymbol{\eta}_{i}\in\mathbb{R}^{dn}$ satisfies $\text{{\bf span}}\left\{ \boldsymbol{\eta}_{1},\ldots,\boldsymbol{\eta}_{r}\right\} =\underset{i\in\underline{M}}{\bigcap}\text{{\bf null}}(L(\mathcal{G}_{i}))$
and 
\[
\boldsymbol{\eta}_{i}^{\top}\boldsymbol{\eta}_{j}=\begin{cases}
1, & i=j\\
0 & i\neq j
\end{cases},\,\forall i,j\in\underline{r}.
\]
}
\end{thm}
\textcolor{black}{}
\begin{proof}
\textcolor{black}{Since $\left[\Phi^{*}(t_{0})\right]^{2}=\Phi^{*}(t_{0})$,
then $\Phi^{*}(t_{0})$ is idempotent and diagonalizable, and the
eigenvalues of $\Phi^{*}(t_{0})$ are $0$ or $1$. We shall first
prove that the eigenvector space corresponding to the eigenvalue $1$
of $\Phi^{*}(t_{0})$ is $\underset{t\geq0}{\bigcap}\text{{\bf null}}(L(t))$.
On the one hand, for any $\boldsymbol{\alpha}\in\text{{\bf span}}\left\{ \boldsymbol{\eta}_{1},\boldsymbol{\eta}_{2},\ldots,\boldsymbol{\eta}_{r}\right\} $,
due to $L(t)\boldsymbol{\alpha}=\boldsymbol{0}$ for any $t\geq t_{0}$,
thus $\Phi(t,t_{0})\boldsymbol{\alpha}=\boldsymbol{\alpha}$. By taking
the limit of $\Phi(t,t_{0})$, it is easy to derive $\Phi^{*}(t_{0})\boldsymbol{\alpha}=\boldsymbol{\alpha}$.
Conversely, for an arbitrary $\boldsymbol{\alpha}$ such that $\Phi^{*}(t_{0})\boldsymbol{\alpha}=\boldsymbol{\alpha}$,
$L(t)\boldsymbol{\alpha}=L(t)\Phi^{*}(t_{0})\boldsymbol{\alpha}=\boldsymbol{0}$.
Therefore, one has $\Phi^{*}(t_{0})\boldsymbol{\eta}_{i}=\boldsymbol{\eta}_{i}$
and $\boldsymbol{\eta}_{i}^{\top}\Phi^{*}(t_{0})=\boldsymbol{\eta}_{i}^{\top}$
for any $i\in\underline{r}$. There exists a matrix $P=\left[\boldsymbol{\eta}_{1},\boldsymbol{\eta}_{2},\ldots,\boldsymbol{\eta}_{r},*,\ldots,*\right]$
}together with its inverse $P^{-1}=\left[\boldsymbol{\eta}_{1},\boldsymbol{\eta}_{2},\ldots,\boldsymbol{\eta}_{r},\star,\ldots,\star\right]^{\top}$\textcolor{black}{such
that 
\begin{align*}
\Phi^{*}(t_{0}) & =P\left[\begin{array}{cc}
I_{r} & 0\\
0 & 0
\end{array}\right]P^{-1}\\
 & =\sum_{i=1}^{r}\boldsymbol{\eta}_{i}\boldsymbol{\eta}_{i}^{\top},
\end{align*}
and one can deduce that
\begin{align*}
\boldsymbol{x}^{*} & =\Phi^{*}(t_{0})\boldsymbol{x}(0)\\
 & =\sum_{i=1}^{r}\boldsymbol{\eta}_{i}\boldsymbol{\eta}_{i}^{\top}\boldsymbol{x}(0)\\
 & =\sum_{i=1}^{r}(\boldsymbol{\eta}_{i}^{\top}\boldsymbol{x}(0))\boldsymbol{\eta}_{i}.
\end{align*}
}
\end{proof}
\begin{rem}
\textcolor{black}{In the }\textbf{\textcolor{black}{Theorem}}\textcolor{black}{{}
\ref{lem:cluster-value}, if $\underset{t\rightarrow\infty}{\text{{\bf lim}}}\boldsymbol{x}(t)=\boldsymbol{x}^{*}$
exists for any initial state $\boldsymbol{x}(t_{0})$, and $\underset{t\geq0}{\bigcap}\text{{\bf null}}(L(t))=\{\boldsymbol{0}\}$,
then the system \eqref{equ:matrix-consensus-overall} achieves the
asymptotic stability. Also, when $\underset{t\geq0}{\bigcap}\text{{\bf null}}(L(t))=\{\boldsymbol{1}_{n}\otimes I_{d}\}$,
the average consensus will be achieved for the system \eqref{equ:matrix-consensus-overall}.}
\end{rem}
\begin{rem}
In the \textbf{Theorem} \ref{lem:cluster-value}, it is assumed that
the switching network $\mathcal{G}(t)$ is constructed from a finite
set of graphs. Here, we shall ask whether or not the conclusion holds
if the switching network $\mathcal{G}(t)$ is constructed from an
infinite set of graphs?\textcolor{red}{{} }\textcolor{black}{To see
this, let us choose, for instance, the multi-agent system $\dot{\boldsymbol{x}}(t)=-\frac{1}{\lfloor t+1\rfloor^{2}}L\boldsymbol{x}(t)$,
where $L$ is the matrix-valued Laplacian matrix of a time-invariant
matrix-weighted network. Now, consider the underlying matrix-weighted
switching network corresponding to the Laplacian matrix $\frac{1}{\lfloor t+1\rfloor^{2}}L$.
One can see that $\text{{\bf lim}}_{t\rightarrow\infty}\boldsymbol{x}(t)=e^{-\frac{\pi^{2}}{6}L}\boldsymbol{x}(0)$.
Here, the convergence value is not only related to the null space
of $L(t)$, but also to the other eigenvectors corresponding t}o the
non-zero eigenvalues. However, if we choose the multi-agent system
$\dot{\boldsymbol{x}}(t)=-\lfloor t+1\rfloor L\boldsymbol{x}(t)$
and $L$ is the same as the above example. Let $\text{{\bf null}}(L)=\text{{\bf span}}\left\{ \boldsymbol{\xi}_{1},\ldots,\boldsymbol{\xi}_{m}\right\} $,
where $\boldsymbol{\xi}_{i}\in\mathbb{R}^{dn}$ satisfies 
\[
\boldsymbol{\xi}_{i}^{\top}\boldsymbol{\xi}_{j}=\begin{cases}
1, & i=j\\
0 & i\neq j
\end{cases},\,\forall i,j\in\underline{m}{\color{magenta}.}
\]
Then, $\text{{\bf lim}}_{t\rightarrow\infty}\boldsymbol{x}(t)=\sum_{i=1}^{m}(\boldsymbol{\xi}_{i}^{\top}\boldsymbol{x}(0))\boldsymbol{\xi}_{i}$.
Therefore, the conclusion in \textbf{Theorem} \ref{lem:cluster-value}
does not always hold if the switching network $\mathcal{G}(t)$ is
constructed from an infinite set of graphs.
\end{rem}
Then, we shall proceed to examine quantitative characterization of
cluster consensus\textcolor{black}{{} achieved on matrix-weighted switching
}networks. In particular, we are intended to establish quantitative
connection between the cluster consensus value and specific properties
of $L(t)$. To this end, we introduce the notion of matrix-weighted
integral network; this notion proves crucial in our subsequent analysis.
In the following discussions, we also assume that the weight matrix
associated with $(i,j)\in\mathcal{E}(t)$ satisfies that either $\text{{\bf sgn}}(A_{ij}(t))\geq0$
or $\text{{\bf sgn}}(A_{ij}(t))\leq0$ for $t\ge0$.
\begin{defn}
\label{def:integral graph} \citet{Pan2021Tac} Let $\mathcal{G}(t)=(\mathcal{V},\mathcal{E}(t),A(t))$
be a matrix-weighted switching network. The matrix-weighted integral
network of $\mathcal{G}(t)$ over time span $[t_{1},t_{2})\subseteq[0,\infty)$
is defined as $\widetilde{\mathcal{G}}_{[t_{1},t_{2})}=(\mathcal{V},\widetilde{\mathcal{E}},\widetilde{A})$,
where 

\[
\widetilde{\mathcal{E}}=\left\{ (i,j)\in\mathcal{V}\times\mathcal{V}\mid\intop_{t_{1}}^{t_{2}}|A_{ij}(t)|dt\succ0\,\text{or}\,\intop_{t_{1}}^{t_{2}}|A_{ij}(t)|dt\succeq0\right\} ,
\]
and
\[
\widetilde{A}=\frac{1}{t_{2}-t_{1}}\intop_{t_{1}}^{t_{2}}A(t)dt.
\]
\end{defn}
According to\textbf{ Definition} \ref{def:integral graph}, let $\widetilde{D}$
denote the matrix-valued degree matrix of $\widetilde{\mathcal{G}}_{[t_{1},t_{2})}$,
that is, 
\[
\widetilde{D}=\frac{1}{t_{2}-t_{1}}\intop_{t_{1}}^{t_{2}}D(t)dt.
\]
Furthermore, let $\widetilde{L}_{[t_{1},t_{2})}$ denote the matrix-valued
Laplacian of $\widetilde{\mathcal{G}}_{[t_{1},t_{2})}$. Thus,

\[
\widetilde{L}_{[t_{1},t_{2})}=\widetilde{D}-\widetilde{A}=\frac{1}{t_{2}-t_{1}}\intop_{t_{1}}^{t_{2}}L(t)dt.
\]

\textcolor{black}{Under the definition of the integral network of
matrix-weighted switching networks, we shall explore connections between
the null space of the matrix-valued Laplacian matrices of a sequence
of matrix-weighted networks and that of the corresponding integral
network, which proves crucial in our subsequent analysis. With reference
to }\textbf{\textcolor{black}{Assumption }}\textcolor{black}{1, we
denote $\mathcal{G}(t)$ on dwell time $t\in[t_{k},t_{k+1})$ as $\mathcal{G}_{[t_{k},t_{k+1})}(t)=\mathcal{G}^{k}$
and denote the associated matrix-valued Laplacian as $L^{k}$, where
$k\in\mathbb{N}$. }
\begin{lem}
\textcolor{black}{\label{thm:null space equality} Let $\mathcal{G}(t)$
be a matrix-weighted switching network satisfying }\textbf{\textcolor{black}{Assumption}}\textcolor{black}{{}
1.}\textbf{\textcolor{black}{{} }}\textcolor{black}{Then
\[
\text{{\bf null}}(\widetilde{L}_{[t_{k^{\prime}},t_{k^{\prime\prime}})})=\underset{i\in\underline{k^{\prime\prime}-k^{\prime}}}{\bigcap}\text{{\bf null}}(L^{k^{\prime}+i-1}),
\]
where $k^{\prime}<k^{\prime\prime}\in\mathbb{N}$.}
\end{lem}
\begin{proof}
\textcolor{black}{On the one hand, we shall prove that $\text{{\bf null}}(\widetilde{L}_{[t_{k^{\prime}},t_{k^{\prime\prime}})})\subseteq\underset{i\in\underline{k^{\prime\prime}-k^{\prime}}}{\bigcap}\text{{\bf null}}(L^{k^{\prime}+i-1})$,
i.e., for any $\boldsymbol{\eta}\in\text{{\bf null}}(\widetilde{L}_{[t_{k^{\prime}},t_{k^{\prime\prime}})})$,
one has $\boldsymbol{\eta}\in\text{{\bf null}}(L^{k^{\prime}+i-1})$
for all $i\in\underline{k^{\prime\prime}-k^{\prime}}$. Note that
$\boldsymbol{\eta}^{\top}\widetilde{L}_{[t_{k^{\prime}},t_{k^{\prime\prime}})}\boldsymbol{\eta}=0$,
implying that,}

\textcolor{black}{
\begin{align*}
 & \boldsymbol{\eta}^{\top}\widetilde{L}_{[t_{k^{\prime}},t_{k^{\prime\prime}})}\boldsymbol{\eta}\\
 & =\boldsymbol{\eta}^{\top}\left(\frac{1}{t_{k^{\prime\prime}}-t_{k^{\prime}}}\intop_{t_{k^{\prime}}}^{t_{k^{\prime\prime}}}L(t)dt\right)\boldsymbol{\eta}\\
 & =\frac{1}{t_{k^{\prime\prime}}-t_{k^{\prime}}}\sum_{i=1}^{k^{\prime\prime}-k^{\prime}}\boldsymbol{\eta}^{\top}L^{k^{\prime}+i-1}(t_{k^{\prime}+i}-t_{k^{\prime}+i-1})\boldsymbol{\eta}\\
 & =\boldsymbol{0},
\end{align*}
due to the fact that $L^{k^{\prime}+i-1}$ is positive semi-definite
or positive definite for all $i\in\underline{k^{\prime\prime}-k^{\prime}}$,
therefore, $\boldsymbol{\eta}^{\top}L^{k^{\prime}+i-1}\boldsymbol{\eta}=0$
for all $i\in\underline{k^{\prime\prime}-k^{\prime}}$, one has $L^{k^{\prime}+i-1}\boldsymbol{\eta}=\boldsymbol{0}$
and $\boldsymbol{\eta}\in\text{{\bf null}}(L^{k^{\prime}+i-1})$ for
all $i\in\underline{k^{\prime\prime}-k^{\prime}}$, which would imply,
\[
\text{{\bf null}}(\widetilde{L}_{[t_{k^{\prime}},t_{k^{\prime\prime}})})\subseteq\underset{i\in\underline{k^{\prime\prime}-k^{\prime}}}{\bigcap}\text{{\bf null}}(L^{k^{\prime}+i-1}).
\]
}

\textcolor{black}{}

\textcolor{black}{On the other hand, we shall prove that $\underset{i\in\underline{k^{\prime\prime}-k^{\prime}}}{\bigcap}\text{{\bf null}}(L^{k^{\prime}+i-1})\subseteq\text{{\bf null}}(\widetilde{L}_{[t_{k^{\prime}},t_{k^{\prime\prime}})})$,
i.e., for any $\boldsymbol{\eta}\in\underset{i\in\underline{k^{\prime\prime}-k^{\prime}}}{\bigcap}\text{{\bf null}}(L^{k^{\prime}+i-1})$,
one has $\boldsymbol{\eta}\in\text{{\bf null}}(\widetilde{L}_{[t_{k^{\prime}},t_{k^{\prime\prime}})})$.
Considering the quantity $\boldsymbol{\eta}^{\top}\widetilde{L}_{[t_{k^{\prime}},t_{k^{\prime\prime}})}\boldsymbol{\eta}$, }

\textcolor{black}{
\begin{align*}
 & \boldsymbol{\eta}^{\top}\widetilde{L}_{[t_{k^{\prime}},t_{k^{\prime\prime}})}\boldsymbol{\eta}\\
 & =\boldsymbol{\eta}^{\top}\left(\frac{1}{t_{k^{\prime\prime}}-t_{k^{\prime}}}\intop_{t_{k^{\prime}}}^{t_{k^{\prime\prime}}}L(t)dt\right)\boldsymbol{\eta}\\
 & =\frac{1}{t_{k^{\prime\prime}}-t_{k^{\prime}}}\sum_{i=1}^{k^{\prime\prime}-k^{\prime}}\boldsymbol{\eta}^{\top}L^{k^{\prime}+i-1}(t_{k^{\prime}+i}-t_{k^{\prime}+i-1})\boldsymbol{\eta}\\
 & =\boldsymbol{0},
\end{align*}
due to the fact that $\widetilde{L}_{[t_{k^{\prime}},t_{k^{\prime\prime}})}$
is positive semi-definite or positive definite, therefore, $\widetilde{L}_{[t_{k^{\prime}},t_{k^{\prime\prime}})}\boldsymbol{\eta}=\boldsymbol{0}$,
which would imply, 
\[
\underset{i\in\underline{k^{\prime\prime}-k^{\prime}}}{\bigcap}\text{{\bf null}}(L^{k^{\prime}+i-1})\subseteq\text{{\bf null}}(\widetilde{L}_{[t_{k^{\prime}},t_{k^{\prime\prime}})}).
\]
Thus,
\[
\text{{\bf null}}(\widetilde{L}_{[t_{k^{\prime}},t_{k^{\prime\prime}})})=\underset{i\in\underline{k^{\prime\prime}-k^{\prime}}}{\bigcap}\text{{\bf null}}(L^{k^{\prime}+i-1}).
\]
}
\end{proof}
\textcolor{black}{}
\textbf{\textcolor{black}{Lemma}}\textcolor{black}{{} \ref{thm:null space equality}
indicates that the intersection of the null space of matrix-valued
Laplacian matrices associated with a sequence of matrix-weighted networks
is equal to the null space of} the corresponding integral network.
Using this fact, we proceed to explo\textcolor{black}{re the sufficient
conditions under which the multi-agent system \eqref{equ:matrix-consensus-overall}
achieves cluster consensus; these conditions reveal the connection
between the  steady-state of the multi-agent system \eqref{equ:matrix-consensus-overall}
and the null} space of the related integral network.

\textcolor{black}{Denote the state transition matrix of multi-agent
system \eqref{equ:matrix-consensus-overall} over time interval $[t_{k^{\prime}},t_{k^{\prime\prime}}]$
as 
\[
\varPhi(t_{k^{\prime\prime}},t_{k^{\prime}})=e^{-L^{k^{\prime\prime}-1}\triangle t_{k^{\prime\prime}-1}}\cdots e^{-L^{k^{\prime}}\triangle t_{k^{\prime}}},
\]
then $\boldsymbol{x}(t_{k^{\prime\prime}})=\varPhi(t_{k^{\prime\prime}},t_{k^{\prime}})\boldsymbol{x}(t_{k^{\prime}})$,
where $k^{\prime}<k^{\prime\prime}\in\mathbb{N}$. }

\textcolor{black}{Let $\lambda_{1}\leq\lambda_{2}\leq\cdots\leq\lambda_{dn}$
be the eigenvalues of the matrix-valued Laplacian matrix $L$ corresponding
to a time-invariant matrix-weighted network. Let $\text{{\bf dim}}(\text{{\bf null}}(L))=m$,
where $m\in\underline{dn}$, namely, 
\[
0=\lambda_{1}=\cdots=\lambda_{m}\leq\lambda_{m+1}\le\cdots\leq\lambda_{dn}.
\]
Denote by $\beta_{1}\geq\beta_{2}\geq\cdots\geq\beta_{dn}$ as the
eigenvalues of $e^{-Lt}$; then $\beta_{i}(e^{-Lt})=e^{-\lambda_{i}(L)t}$,
i.e., $1=\beta_{1}=\cdots=\beta_{m}\geq\beta_{m+1}\geq\cdots\geq\beta_{dn}$.
In the meantime, the eigenvector corresponding to the eigenvalue $\beta_{i}(e^{-Lt})$
is equal to that corresponding to $\lambda_{i}(L)$. Therefore, let
\[
\text{{\bf dim}}(\underset{i\in\underline{k^{\prime\prime}-k^{\prime}}}{\bigcap}\text{{\bf null}}(L^{k^{\prime}+i-1}))=m,
\]
where $m\in\underline{dn}$, then $\varPhi(t_{k^{\prime\prime}},t_{k^{\prime}})^{\top}\varPhi(t_{k^{\prime\prime}},t_{k^{\prime}})$
has at least $m$ eigenvalues at $1$. Let $\mu_{j}$ be the eigenvalues
of $\varPhi(t_{k^{\prime\prime}},t_{k^{\prime}})^{\top}\varPhi(t_{k^{\prime\prime}},t_{k^{\prime}})$,
where $j\in\underline{dn}$ such that $\mu_{1}=\cdots=\mu_{m}=1$
and $\mu_{m+1}\geq\mu_{m+2}\geq\cdots\geq\mu_{dn}$. Then applying
the facts that $\varPhi(t_{k^{\prime\prime}},t_{k^{\prime}})^{\top}\varPhi(t_{k^{\prime\prime}},t_{k^{\prime}})\ge0$
and 
\begin{align*}
 & \underset{j\in\underline{dn}}{\text{{\bf max}}}\mu_{j}(\varPhi(t_{k^{\prime\prime}},t_{k^{\prime}})^{\top}\varPhi(t_{k^{\prime\prime}},t_{k^{\prime}}))\\
 & =\parallel\varPhi(t_{k^{\prime\prime}},t_{k^{\prime}})^{\top}\varPhi(t_{k^{\prime\prime}},t_{k^{\prime}})\parallel\\
 & \leq1,
\end{align*}
one has $\mu_{dn}\leq\cdots\leq\mu_{m+2}\leq\mu_{m+1}\leq1$. The
following lemma thereby provides the relationship between the null
space of the matrix-valued Laplacian of $\widetilde{\mathcal{G}}_{[t_{k^{\prime}},t_{k^{\prime\prime}})}$
and the eigenvalues of $\varPhi(t_{k^{\prime\prime}},t_{k^{\prime}})^{\top}\varPhi(t_{k^{\prime\prime}},t_{k^{\prime}})$,
which is paramount in the subsequent analysis.}
\begin{lem}
\textcolor{black}{\label{lem:eigenvalue inequality}Let $\mathcal{G}(t)$
be a matrix-weighted switching network satisfying Assumption 1. Let
$\text{{\bf dim}}(\text{{\bf null}}(\widetilde{L}_{[t_{k^{\prime}},t_{k^{\prime\prime}})}))=m$,
where $m\in\underline{dn}$. Then
\[
\mu_{m+1}(\varPhi(t_{k^{\prime\prime}},t_{k^{\prime}})^{\top}\varPhi(t_{k^{\prime\prime}},t_{k^{\prime}}))<1,
\]
where $k^{\prime}<k^{\prime\prime}\in\mathbb{N}$.}
\end{lem}
\begin{proof}
\textcolor{black}{By contradiction, assume that 
\[
\mu_{m+1}(\varPhi(t_{k^{\prime\prime}},t_{k^{\prime}})^{\top}\varPhi(t_{k^{\prime\prime}},t_{k^{\prime}}))=1
\]
for $k^{\prime}<k^{\prime\prime}\in\mathbb{N}$. According to }\textbf{\textcolor{black}{Lemma}}\textcolor{black}{{}
\ref{lem:Rayleigh Theorem}, there exists a non-zero $\boldsymbol{\eta}\notin\text{{\bf null}}(\widetilde{L}_{[t_{k^{\prime}},t_{k^{\prime\prime}})})$
such that
\[
\parallel\boldsymbol{\eta}\parallel=\parallel\varPhi(t_{k^{\prime\prime}},t_{k^{\prime}})\boldsymbol{\eta}\parallel.
\]
}

\textcolor{black}{Denote $\boldsymbol{\eta}_{k^{\prime}}=\boldsymbol{\eta}$
and $\boldsymbol{\eta}_{k^{\prime}+i}=e^{-L^{k^{\prime}+i-1}\triangle t_{k^{\prime}+i-1}}\boldsymbol{\eta}_{k^{\prime}+i-1}$
for all $i\in\underline{k^{\prime\prime}-k^{\prime}}$. Moreover,
$\lambda_{j}(e^{-L^{k^{\prime}+i-1}\triangle t_{k^{\prime}+i-1}})\leq1$
for all $j\in\underline{dn}$, thereby, 
\[
\parallel e^{-L^{k^{\prime}+i-1}\triangle t_{k^{\prime}+i-1}}\boldsymbol{\eta}_{k^{\prime}+i-1}\parallel\leq\parallel\boldsymbol{\eta}_{k^{\prime}+i-1}\parallel,
\]
and  
\begin{align*}
\parallel\boldsymbol{\eta}\parallel & =\parallel\boldsymbol{\eta}_{k^{\prime\prime}}\parallel\leq\parallel\boldsymbol{\eta}_{k^{\prime\prime}-1}\parallel\leq\ldots\leq\parallel\boldsymbol{\eta}_{k^{\prime}}\parallel=\parallel\boldsymbol{\eta}\parallel.
\end{align*}
Hence, 
\[
\parallel e^{-L^{k^{\prime}+i-1}\triangle t_{k^{\prime}+i-1}}\boldsymbol{\eta}_{k^{\prime}+i-1}\parallel=\parallel\boldsymbol{\eta}_{k^{\prime}+i-1}\parallel.
\]
 Then 
\begin{align*}
 & \boldsymbol{\eta}_{k^{\prime}+i-1}^{\top}e^{-L^{k^{\prime}+i-1}\triangle t_{k^{\prime}+i-1}}e^{-L^{k^{\prime}+i-1}\triangle t_{k^{\prime}+i-1}}\boldsymbol{\eta}_{k^{\prime}+i-1}\\
= & \boldsymbol{\eta}_{k^{\prime}+i-1}^{\top}\boldsymbol{\eta}_{k^{\prime}+i-1}.
\end{align*}
By }\textbf{\textcolor{black}{Lemma}}\textcolor{black}{{} \ref{lem:Rayleigh Theorem},
\[
e^{-2L^{k^{\prime}+i-1}\triangle t_{k^{\prime}+i-1}}\boldsymbol{\eta}_{k^{\prime}+i-1}=\boldsymbol{\eta}_{k^{\prime}+i-1},
\]
and thus,
\[
L^{k^{\prime}+i-1}\boldsymbol{\eta}_{k^{\prime}+i-1}=\boldsymbol{0},
\]
implying that $\boldsymbol{\eta}_{k^{\prime}+i-1}\in\text{{\bf null}}(L^{k^{\prime}+i-1})$.
Using the fact 
\begin{align*}
\parallel\boldsymbol{\eta}_{k^{\prime}+i} & -\boldsymbol{\eta}_{k^{\prime}+i-1}\parallel\\
 & =\parallel e^{-L^{k^{\prime}+i-1}\triangle t_{k^{\prime}+i-1}}\boldsymbol{\eta}_{k^{\prime}+i-1}-\boldsymbol{\eta}_{k^{\prime}+i-1}\parallel\\
 & =\parallel\sum_{t=1}^{\infty}\frac{1}{t!}(-L^{k^{\prime}+i-1}\triangle t_{k^{\prime}+i-1})^{t}\boldsymbol{\eta}_{k^{\prime}+i-1}\parallel\\
 & =0,
\end{align*}
one can conclude that $\boldsymbol{\eta}_{k^{\prime}+i-1}=\boldsymbol{\eta}_{k^{\prime}+i}$
for all $i\in\underline{k^{\prime\prime}-k^{\prime}}$, which implies
that $\boldsymbol{\eta}\in\underset{i\in\underline{k^{\prime\prime}-k^{\prime}}}{\cap}\text{{\bf null}}(L^{k^{\prime}+i-1})$,
i.e., $\boldsymbol{\eta}\in\text{{\bf null}}(\widetilde{L}_{[t_{k^{\prime}},t_{k^{\prime\prime}})})$,
leading to a contradiction.}
\end{proof}
\textcolor{black}{Based on the above established Lemmas, we shall
show the main result of this part using null space analysis of matrix-valued
Laplacian related of integral network associated with the switching
networks.}

\begin{thm}
\textcolor{black}{\label{thm:cluster theorem}Let ${\normalcolor \mathcal{G}(t)}$
be a matrix-weighted switching network satisfying }\textbf{\textcolor{black}{Assumption}}\textcolor{black}{{}
1. If there exists a subsequence of $\{t_{k}|k\in\mathbb{N}\}$, denoted
by 
\[
\{t_{k_{l}}|t_{k_{0}}=t_{0},\triangle t_{k_{l}}=t_{k_{l+1}}-t_{k_{l}}<\infty,l\in\mathbb{N}\},
\]
and a scalar $q\in(0,1)$, such that for all $l\in\mathbb{N}$, 
\[
\text{{\bf null}}(\widetilde{L}_{[t_{k_{l}},t_{k_{l+1}})})=\text{{\bf null}}(\widetilde{L}_{[t_{k_{l+1}},t_{k_{l+2}})}),
\]
 and 
\[
\mu_{m+1}(\varPhi(t_{k_{l+1}},t_{k_{l}})^{\top}\varPhi(t_{k_{l+1}},t_{k_{l}}))\leq q,
\]
where $m=\text{{\bf dim}}(\text{{\bf null}}(\widetilde{L}_{[t_{k_{l}},t_{k_{l+1}})}))$.
Then the multi-agent network \eqref{equ:matrix-consensus-overall}
admits the cluster consensus. Moreover, denote $\text{{\bf null}}(\widetilde{L}_{[t_{k_{l}},t_{k_{l+1}})})=\text{{\bf span}}\{\boldsymbol{\xi}_{1},\cdots,\boldsymbol{\xi}_{m}\}$
for all $l\in\mathbb{N}$, where $\boldsymbol{\xi}_{i}\in\mathbb{R}^{dn}$
satisfies 
\[
\boldsymbol{\xi}_{i}^{\top}\boldsymbol{\xi}_{j}=\begin{cases}
1, & i=j\\
0 & i\neq j
\end{cases},\,\forall i,j\in\underline{m},
\]
then the cluster consensus value is 
\[
\boldsymbol{x}^{*}=\sum_{i=1}^{m}(\boldsymbol{\xi}_{i}^{\top}\boldsymbol{x}(0))\boldsymbol{\xi}_{i}.
\]
}
\end{thm}
\begin{proof}
\textcolor{black}{Construct the error vector $\boldsymbol{\omega}(t)=\boldsymbol{x}(t)-\boldsymbol{x}^{*}$
which satisfies that $\dot{\boldsymbol{\omega}}(t)=-L(t)\boldsymbol{\omega}(t)$.
Choose $\boldsymbol{\omega}(0)\notin\text{{\bf null}}(\widetilde{L}_{[t_{k_{l}},t_{k_{l+1}})})$
for any $l\in\mathbb{N}$, then $\boldsymbol{\omega}(0)^{\top}\boldsymbol{\xi}_{i}=0$
for all $i\in\underline{m}$. Thus, $\boldsymbol{\omega}(0)\bot\text{{\bf null}}(\widetilde{L}_{[t_{k_{l}},t_{k_{l+1}})})$
for any $l\in\mathbb{N}$. Applying }\textbf{\textcolor{black}{Lemma}}\textcolor{black}{{}
\ref{lem:Rayleigh Theorem} yields
\[
\mu_{m+1}(\varPhi(t_{k_{1}},t_{k_{0}})^{\top}\varPhi(t_{k_{1}},t_{k_{0}}))\geq\frac{\boldsymbol{\omega}(t_{k_{1}})^{\top}\boldsymbol{\omega}(t_{k_{1}})}{\boldsymbol{\omega}(0)^{\top}\boldsymbol{\omega}(0)},
\]
}implying that,
\[
\parallel\boldsymbol{\omega}(t_{k_{1}})\parallel\leq\mu_{m+1}(\varPhi(t_{k_{1}},t_{k_{0}})^{\top}\varPhi(t_{k_{1}},t_{k_{0}}))^{\frac{1}{2}}\parallel\boldsymbol{\omega}(0)\parallel.
\]
Therefore, for any $l\in\mathbb{Z}_{+}$
\begin{align*}
 & \parallel\boldsymbol{\omega}(t_{k_{l+1}})\parallel\\
\leq & \left({\displaystyle \prod_{s=0}^{l}}\mu_{m+1}(\varPhi(t_{k_{s+1}},t_{k_{s}})^{\top}\varPhi(t_{k_{s+1}},t_{k_{s}}))^{\frac{1}{2}}\right)\parallel\boldsymbol{\omega}(0)\parallel\\
\leq & q^{\frac{1}{2}(l+1)}\parallel\boldsymbol{\omega}(0)\parallel.
\end{align*}
Let 
\[
V(t)=\boldsymbol{\omega}(t)^{\top}\boldsymbol{\omega}(t)=\parallel\boldsymbol{\omega}(t)\parallel^{2};
\]
then computing the derivative of $V(t)$ along the trajectories of
system $\dot{\boldsymbol{\omega}}(t)=-L(t)\boldsymbol{\omega}(t)$
yields,
\[
\dot{V}(t)=2\boldsymbol{\omega}(t)^{\top}(-L(t))\boldsymbol{\omega}(t)\leq0.
\]
 Thus
\[
\parallel\boldsymbol{\omega}(t)\parallel\leq\parallel\boldsymbol{\omega}(t_{k_{l}})\parallel\leq q^{\frac{1}{2}l}\parallel\boldsymbol{\omega}(0)\parallel,
\]
for any $t\in[t_{k_{l}},t_{k_{l+1}})$ and $l\in\mathbb{N}$. Note
that $0<q<1$, and hence,
\[
{\displaystyle \lim_{t\rightarrow\infty}}\parallel\boldsymbol{\omega}(t)\parallel=0.
\]
As such, the multi-agent system  \eqref{equ:matrix-consensus-overall}
achieves cluster consensus and the\textcolor{magenta}{{} }\textcolor{black}{cluster
consensus valu}e is $\boldsymbol{x}^{*}=\sum_{i=1}^{m}(\boldsymbol{\xi}_{i}^{\top}\boldsymbol{x}(0))\boldsymbol{\xi}_{i}$.
\end{proof}
\begin{rem}
\textcolor{black}{In the }\textbf{\textcolor{black}{Theorem}}\textcolor{black}{{}
\ref{thm:cluster theorem}, for a matrix-weighted switching network
${\normalcolor \mathcal{G}(t)}$, if there exists a subsequence $\{t_{h_{l}}|t_{h_{0}}=t_{0},\triangle t_{h_{l}}=t_{h_{l+1}}-t_{h_{l}}<\infty,l\in\mathbb{N}\}$
of $\{t_{k}|k\in\mathbb{N}\}$ such that $\text{{\bf null}}(\widetilde{L}_{[t_{h_{l}},t_{h_{l+1}})})=\text{{\bf null}}(\widetilde{L}_{[t_{h_{l+1}},t_{h_{l+2}})})$
for any $l\in\mathbb{N}$, then there does not exist another subsequence
$\{t_{q_{l}}|t_{q_{0}}=t_{0},\triangle t_{q_{l}}=t_{q_{l+1}}-t_{q_{l}}<\infty,l\in\mathbb{N}\}$
of $\{t_{k}|k\in\mathbb{N}\}$ such that $\text{{\bf null}}(\widetilde{L}_{[t_{q_{l}},t_{q_{l+1}})})=\text{{\bf null}}(\widetilde{L}_{[t_{q_{l+1}},t_{q_{l+2}})})$
for any $l\in\mathbb{N}$ and $\text{{\bf null}}(\widetilde{L}_{[t_{h_{l}},t_{h_{l+1}})})\neq\text{{\bf null}}(\widetilde{L}_{[t_{q_{l}},t_{q_{l+1}})})$
for any $l\in\mathbb{N}$. }We shall illustrate this point by contradiction.
Choose one time interval $[t_{m},t_{n})$ where $m<n\in\mathbb{N}$,
such that there exist $l_{0}\in\mathbb{N}$ and $h_{0}\in\mathbb{N}$
satisfying $[t_{k_{l_{0}}},t_{k_{l_{0}+1}})\subseteq[t_{m},t_{n})$
and $[t_{k_{h_{0}}},t_{k_{h_{0}+1}})\subseteq[t_{m},t_{n})$, then
one has $\text{{\bf null}}(\widetilde{L}_{[t_{m},t_{n})})=\text{{\bf null}}(\widetilde{L}_{[t_{k_{l}},t_{k_{l+1}})})$
for any $l\in\mathbb{N}$ and $\text{{\bf null}}(\widetilde{L}_{[t_{m},t_{n})})=\text{{\bf null}}(\widetilde{L}_{[t_{k_{h}},t_{k_{h+1}})})$
for any $h\in\mathbb{N}$; however, $\text{{\bf null}}(\widetilde{L}_{[t_{k_{l}},t_{k_{l+1}})})\neq\text{{\bf null}}(\widetilde{L}_{[t_{k_{h}},t_{k_{h+1}})})$
for any $l\in\mathbb{N}$ and $h\in\mathbb{N}$, which is a contradiction.
\end{rem}
\begin{rem}
\textcolor{black}{Consider a special class of switching networks,
where $\mathcal{G}(t)$ is periodic, i.e., there exists a $T>0$ such
that $\mathcal{G}(t+T)=\mathcal{G}(t)$ for any $t\geq0$. One can
see that it satisfies the condition of }\textbf{\textcolor{black}{Theorem}}\textcolor{black}{{}
\ref{thm:cluster theorem}, therefore, one can apply }\textbf{\textcolor{black}{Theorem}}\textcolor{black}{{}
3 in \citet{trinh2018matrix} on the integral network of $\mathcal{G}(t)$
over one period to derive the cluster situation for switching networks
$\mathcal{G}(t)$.}
\end{rem}
\textcolor{black}{Bipartite consensus is a special case of cluster
consensus in the scalar-weighted time-invariant signed networks. Different
from the scalar-weighted time-invariant signed networks where the
connectivity and the structurally balance of the network can completely
guarantee the bipartite consensus, for the matrix-weighted time-in}variant
signed networks, even if the network is unbalanced, there may be a
bipartite consensus solution. Recently, authors in \citet{su2019bipartite}
provide a necessary and sufficient condition for achieving bipartite
consensus from an algebraic perspective, that is, the null space of
the matrix-valued Laplacian matrix corresponding to the matrix-weighted
signed networks is in the form of $C(1_{n}\otimes\Psi)$, where $\Psi=[\boldsymbol{\varphi}_{1},\boldsymbol{\varphi}_{2},\ldots,\boldsymbol{\varphi}_{m}],\,m\in\mathbb{Z}_{+}$
and $\boldsymbol{\varphi}_{i}\in\mathbb{R}^{d}$, $i\in\underline{m}$,
are mutually perpendicular unit basis vectors, $C=\text{{\bf diag}}\left\{ \sigma_{1},\sigma_{2},\ldots,\sigma_{n}\right\} \in\mathbb{R}^{dn\times dn}$
and $\sigma_{i}=I_{d}$ or $\sigma_{i}=-I_{d}$. Based on these results,
next we shall examine conditions for the bipartite consensus u\textcolor{black}{nder
the matrix-weighted switching networks.}
\begin{cor}
\textcolor{black}{\label{thm: bipartite-consensus-condition} Let
${\normalcolor \mathcal{G}(t)}$ be a matrix-weighted switching network
satisfying }\textbf{\textcolor{black}{Assumption }}\textcolor{black}{1;
furthermore, suppose there exists a subsequence of $\{t_{k}|k\in\mathbb{N}\}$,
denoted by $\{t_{k_{l}}|t_{k_{0}}=t_{0},\triangle t_{k_{l}}=t_{k_{l+1}}-t_{k_{l}}<\infty,l\in\mathbb{N}\}$,
and a scalar $q\in(0,1)$, such that $\text{{\bf null}}(\widetilde{L}_{[t_{k_{l}},t_{k_{l+1}})})=C(1_{n}\otimes\Psi)$
and $\mu_{m+1}(\varPhi(t_{k_{l+1}},t_{k_{l}})^{\top}\varPhi(t_{k_{l+1}},t_{k_{l}}))\leq q$
for all $l\in\mathbb{N}$, where $\Psi=[\boldsymbol{\varphi}_{1},\boldsymbol{\varphi}_{2},\ldots,\boldsymbol{\varphi}_{m}],\,m\in\mathbb{Z}_{+}$
and $\boldsymbol{\varphi}_{i}\in\mathbb{R}^{d}$ is the unit basis
vector and vertical to each other for all $i\in\underline{m}$. Then
the multi-agent network \eqref{equ:matrix-consensus-overall} admits
the bipartite consensus, and the bipartite consensus value is 
\[
\boldsymbol{x}^{*}=C\left(\boldsymbol{1}_{n}\otimes\left(\frac{1}{n}\Psi\left(\boldsymbol{1}_{n}^{\top}\otimes\Psi^{\top}\right)C\boldsymbol{x}(t_{0})\right)\right).
\]
}
\end{cor}
\begin{proof}
\textcolor{black}{The process is similar to the proof of }\textbf{\textcolor{black}{Theorem}}\textcolor{black}{{}
\ref{thm:cluster theorem}, thus we omit here.}
\end{proof}
\begin{rem}
\textcolor{black}{Notably, the number of candidate networks for switching
and dwell times in the aforementioned discussions can be infinite,
which implies that $\left\{ \varPhi(t_{k_{l+1}},t_{k_{l}})^{\top}\varPhi(t_{k_{l+1}},t_{k_{l}})\,|\,l\in\mathbb{N}\right\} $
cannot be generated from a finite set. Therefore, in }\textbf{\textcolor{black}{Corollary}}\textcolor{black}{{}
\ref{thm: bipartite-consensus-condition}, condition $\mu_{m+1}(\varPhi(t_{k_{l+1}},t_{k_{l}})^{\top}\varPhi(t_{k_{l+1}},t_{k_{l}}))\leq q$
is used to ensure bipartite consensus. Subsequently, in order to remove
this condition and obtain the analogous graph-theoretic condition
for reaching bipartite consensus, we proceed to discuss the case where
both the switching networks and the dwell times come from a finite
set \citet{cao2008reaching,ren2005consensus}}\textbf{\textcolor{black}{.}}
\end{rem}
\textcolor{black}{}
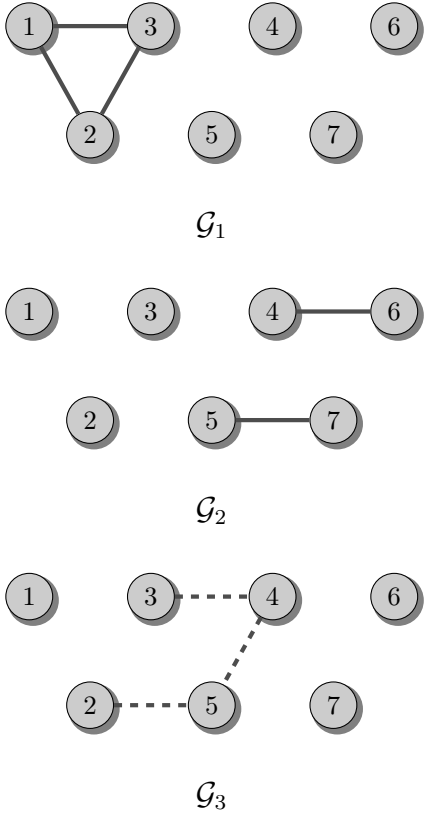
\begin{figure}
\begin{centering}
\textcolor{black}{\begin{tikzpicture}[scale=0.8]
    \node (n1) at (-3,0.3) [circle,circular drop shadow,fill=black!20,draw] {1};
    \node (n2) at (-2,-1.5) [circle,circular drop shadow,fill=black!20,draw] {2};
    \node (n3) at (-1 ,0.3) [circle,circular drop shadow,fill=black!20,draw] {3};
    \node (n4) at (1 ,0.3) [circle,circular drop shadow,fill=black!20,draw] {4};
	\node (n5) at (0,-1.5) [circle,circular drop shadow,fill=black!20,draw] {5};
    \node (n6) at (3 ,0.3) [circle,circular drop shadow,fill=black!20,draw] {6};
    \node (n7) at (2,-1.5) [circle,circular drop shadow,fill=black!20,draw] {7};
    \node (G1) at (0,-3) {\large{$\mathcal{G}_1$}};
	\draw[-, ultra thick, color=black!70] [-] (n1) -- (n2); 
	\draw [-, ultra thick, color=black!70] (n1) -- (n3); 
    \draw [-, ultra thick, color=black!70] (n2) -- (n3); 
\end{tikzpicture}\vskip 0.5cm\begin{tikzpicture}[scale=0.8]
    \node (n1) at (-3,0.3) [circle,circular drop shadow,fill=black!20,draw] {1};
    \node (n2) at (-2,-1.5) [circle,circular drop shadow,fill=black!20,draw] {2};
    \node (n3) at (-1 ,0.3) [circle,circular drop shadow,fill=black!20,draw] {3};
    \node (n4) at (1 ,0.3) [circle,circular drop shadow,fill=black!20,draw] {4};
	\node (n5) at (0,-1.5) [circle,circular drop shadow,fill=black!20,draw] {5};
    \node (n6) at (3 ,0.3) [circle,circular drop shadow,fill=black!20,draw] {6};
    \node (n7) at (2,-1.5) [circle,circular drop shadow,fill=black!20,draw] {7};
    \node (G2) at (0,-3) {\large{$\mathcal{G}_2$}};
	\draw[-, ultra thick, color=black!70] (n4) -- (n6); 
	\draw[-, ultra thick, color=black!70] (n5) -- (n7); 
\end{tikzpicture}\vskip 0.5cm\begin{tikzpicture}[scale=0.8]	
    \node (n1) at (-3,0.3) [circle,circular drop shadow,fill=black!20,draw] {1};
    \node (n2) at (-2,-1.5) [circle,circular drop shadow,fill=black!20,draw] {2};
    \node (n3) at (-1 ,0.3) [circle,circular drop shadow,fill=black!20,draw] {3};
    \node (n4) at (1 ,0.3) [circle,circular drop shadow,fill=black!20,draw] {4};
	\node (n5) at (0,-1.5) [circle,circular drop shadow,fill=black!20,draw] {5};
    \node (n6) at (3 ,0.3) [circle,circular drop shadow,fill=black!20,draw] {6};
    \node (n7) at (2,-1.5) [circle,circular drop shadow,fill=black!20,draw] {7};
    \node (G3) at (0, -3) {\large{$\mathcal{G}_3$}};
	\draw[-, ultra thick, dashed, color=black!70] (n3) -- (n4); 
    \draw[-, ultra thick, dashed, color=black!70] (n4) -- (n5); 
    \draw[-, ultra thick, dashed, color=black!70] (n2) -- (n5); 	
\end{tikzpicture}}
\par\end{centering}
\textcolor{black}{\caption{Three matrix-weighted networks $\mathcal{G}_{1}$, $\mathcal{G}_{2}$
and $\mathcal{G}_{3}$. \textcolor{black}{Those edges weighted by
positive definite matrices are illustrated by solid lines and edges
weighted by positive semi-definite matrices are illustrated by dotted
lines.}}
\label{fig:example-network-1}}
\end{figure}
\textbf{\textcolor{black}{Assumption 3. }}\textcolor{black}{In addition
to }\textbf{\textcolor{black}{Assumption }}\textcolor{black}{1 and
}\textbf{\textcolor{black}{Assumption }}\textcolor{black}{2, the dwell
time $\triangle t_{k}=t_{k+1}-t_{k}$ ($k\in\mathbb{N}$) is chosen
from a finite set of arbitrary positive numbers.}
\begin{defn}[Simultaneously Structurally Balanced]
\textcolor{black}{ A matrix-weighted switching network $\mathcal{G}(t)$
is simultaneously structurally balanced if there exists a time-invariant
bipartition of the node set $\mathcal{V}$, say $\mathcal{V}_{1}$
and $\mathcal{V}_{2}$, such that the matrix weights on the edges
within each subset is positive definite or positive semi-definite,
but negative definite or negative semi-definite for the edges between
the two subsets. A matrix-weighted switching network is simultaneously
structurally imbalanced if it is not simultaneously structurally balanced.}
\end{defn}
\textcolor{black}{On the basis of the above discussions, an analogous
graph-theoretic condition by use of simultaneously structurally balance
is as }follows.
\begin{cor}
\label{thm:structure-condition-bipartite-consensus} Let ${\normalcolor \mathcal{G}(t)}$
be a matrix-weighted simultaneously structurally balanced switching
network satisfying Assumption 3 with a time-invariant node set bipartition
$\mathcal{V}_{1}$ and $\mathcal{V}_{2}$; if there exists a subsequence
of $\{t_{k}|k\in\mathbb{N}\}$, denoted by $\{t_{k_{l}}|t_{k_{0}}=t_{0},\forall l\in\mathbb{N}\}$,
and $h>0$ such that $\triangle t_{k_{l}}=t_{k_{l+1}}-t_{k_{l}}\leq h$
and  the integral graph of $\mathcal{G}(t)$ over time span $[t_{k_{l}},t_{k_{l+1}})$
has a positive-negative spanning tree for all $l\in\mathbb{N}$, then
the multi-agent network \eqref{equ:matrix-consensus-overall} admits
the bipartite consensus, and the bipartite consensus value is 
\[
\boldsymbol{x}^{*}=C\left(\boldsymbol{1}_{n}\otimes\left(\frac{1}{n}\left(\boldsymbol{1}_{n}^{\top}\otimes I_{d}\right)C\boldsymbol{x}(0)\right)\right),
\]
where $C=\text{{\bf diag}}\left\{ \sigma_{1},\sigma_{2},\ldots,\sigma_{n}\right\} \in\mathbb{R}^{dn\times dn}$
satisfies $\sigma_{i}=I_{d}$ if $i\in\mathcal{V}_{1}$ and $\sigma_{i}=-I_{d}$
if $i\in\mathcal{V}_{2}$. 
\end{cor}
\begin{proof}
\textcolor{black}{Since ${\normalcolor \mathcal{G}(t)}$ is simultaneously
structurally balanced, then the integral graph of $\mathcal{G}(t)$
over time span $[t_{k_{l}},t_{k_{l+1}})$ for any $l\in\mathbb{N}$
is structurally balanced. In addition, the integral graph of $\mathcal{G}(t)$
over time span $[t_{k_{l}},t_{k_{l+1}})$ has a positive-negative
spanning tree for all $l\in\mathbb{N}$. Therefore, according to the
}\textbf{\textcolor{black}{Theorem}}\textcolor{black}{{} 2 in \citet{pan2018bipartite},
$\text{{\bf null}}(\widetilde{L}_{[t_{k_{l}},t_{k_{l+1}})})=C(1_{n}\otimes I_{d})$,
one can conclude that $\mu_{d+1}(\varPhi(t_{k_{l+1}},t_{k_{l}})^{\top}\varPhi(t_{k_{l+1}},t_{k_{l}}))<1$
for all $l\in\mathbb{N}$ by }\textbf{\textcolor{black}{Lemma}}\textcolor{black}{{}
\ref{lem:eigenvalue inequality}. Note that }\textbf{\textcolor{black}{Assumption
}}\textcolor{black}{3 ensures that $\left\{ \varPhi(t_{k_{l+1}},t_{k_{l}})^{\top}\varPhi(t_{k_{l+1}},t_{k_{l}})\,|\,l\in\mathbb{N}\right\} $
can be generated from a finite set. Now choose 
\[
q=\underset{l\in\mathbb{N}}{\text{{\bf max}}}\left\{ \mu_{d+1}(\varPhi(t_{k_{l+1}},t_{k_{l}})^{\top}\varPhi(t_{k_{l+1}},t_{k_{l}}))\right\} ;
\]
hence according to the proof of }\textbf{\textcolor{black}{Theorem}}\textcolor{black}{{}
\ref{thm:cluster theorem}, the multi-agent system \eqref{equ:matrix-consensus-overall}
admits bipartite consensus.}
\end{proof}

\section{\textcolor{black}{Simulation Results \label{sec:Simulation-Results}}}

\textcolor{black}{Consider a sequence of matrix-weighted networks,
consisting of (the same) seven agents, where their interaction networks
are $\mathcal{G}_{1},\mathcal{G}_{2}$ and $\mathcal{G}_{3}$, respectively,
as shown in Figure \ref{fig:example-network-1}. Note that $n=7$
and $d=3$ in this example. The matrix-valued edge weights for each
network are, 
\[
A_{12}(\mathcal{G}_{1})=\left[\begin{array}{ccc}
2 & -1 & -1\\
-1 & 3 & -1\\
-1 & -1 & 2
\end{array}\right],\thinspace\thinspace A_{13}(\mathcal{G}_{1})=\left[\begin{array}{ccc}
1 & 1 & 0\\
1 & 2 & 0\\
0 & 0 & 3
\end{array}\right],
\]
\[
A_{23}(\mathcal{G}_{1})=\left[\begin{array}{ccc}
3 & 1 & -1\\
1 & 2 & -1\\
-1 & -1 & 2
\end{array}\right],\thinspace\thinspace A_{46}(\mathcal{G}_{2})=\left[\begin{array}{ccc}
1 & 1 & 0\\
1 & 2 & 0\\
0 & 0 & 3
\end{array}\right],
\]
}

\textcolor{black}{
\[
A_{57}(\mathcal{G}_{2})=\left[\begin{array}{ccc}
2 & -1 & -1\\
-1 & 3 & -1\\
-1 & -1 & 2
\end{array}\right],\thinspace\thinspace A_{34}(\mathcal{G}_{3})=\left[\begin{array}{ccc}
4 & 0 & -2\\
0 & 1 & 1\\
-2 & 1 & 2
\end{array}\right],
\]
\[
\thinspace\thinspace A_{25}(\mathcal{G}_{3})=\left[\begin{array}{ccc}
4 & 2 & 0\\
2 & 2 & 1\\
0 & 1 & 1
\end{array}\right],\thinspace\thinspace A_{45}(\mathcal{G}_{3})=\left[\begin{array}{ccc}
4 & 2 & 0\\
2 & 4 & 3\\
0 & 3 & 3
\end{array}\right].
\]
}

\textcolor{black}{Consider a time sequence $\{t_{k}\thinspace|\thinspace k\in\mathbb{N}\}$
such that $t_{k}=k\Delta t$ where $\Delta t>0$. The coordination
process is initiated from network $\mathcal{G}_{1}$ (i.e., ${\normalcolor \mathcal{G}(0)=\mathcal{G}_{1}}$)
with 
\[
\boldsymbol{x}_{1}(0)=[0.3922,\,0.6555,\,0.1712]^{\top},
\]
\[
\boldsymbol{x}_{2}(0)=[0.7060,\,0.0318,\,0.5762]^{\top},
\]
\[
\boldsymbol{x}_{3}(0)=[0.2688,\,0.1592,\,0.3266]^{\top},
\]
\[
\boldsymbol{x}_{4}(0)=[0.6787,\,0.7577,\,0.7431]^{\top},
\]
\[
\boldsymbol{x}_{5}(0)=[0.3830,\,0.6112,\,0.1212]^{\top},
\]
\[
\boldsymbol{x}_{6}(0)=[0.3555,\,0.9712,\,0.8060]^{\top},
\]
and
\[
\boldsymbol{x}_{7}(0)=[0.1318,\,0.7762,\,0.3688]^{\top}.
\]
}

\textcolor{black}{The switching among networks $\mathcal{G}_{1},\mathcal{G}_{2}$
and $\mathcal{G}_{3}$ satisfies,}

\textcolor{black}{
\begin{equation}
{\normalcolor \mathcal{G}(t)}=\begin{cases}
\begin{array}{c}
\mathcal{G}_{1},\\
\mathcal{G}_{2},\\
\mathcal{G}_{3},
\end{array} & \begin{array}{c}
t\in[t_{6l},t_{6l+2}),\\
t\in[t_{6l+2},t_{6l+5}),\\
t\in[t_{6l+5},t_{6(l+1)}),
\end{array}\end{cases}\label{eq:sp-2}
\end{equation}
where $l\in\mathbb{N}$. }
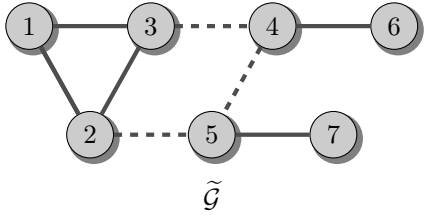
\begin{figure}
\begin{centering}
\textcolor{black}{\begin{tikzpicture}[scale=0.8]
	\node (n1) at (-3,0.3) [circle,circular drop shadow,fill=black!20,draw] {1};
    \node (n2) at (-2,-1.5) [circle,circular drop shadow,fill=black!20,draw] {2};
    \node (n3) at (-1 ,0.3) [circle,circular drop shadow,fill=black!20,draw] {3};
    \node (n4) at (1 ,0.3) [circle,circular drop shadow,fill=black!20,draw] {4};
	\node (n5) at (0,-1.5) [circle,circular drop shadow,fill=black!20,draw] {5};
    \node (n6) at (3 ,0.3) [circle,circular drop shadow,fill=black!20,draw] {6};
    \node (n7) at (2,-1.5) [circle,circular drop shadow,fill=black!20,draw] {7};
    \node (G2) at (0,-2.5) {{$\widetilde{\mathcal{G}}$}};
	\draw[-, ultra thick, color=black!70] (n1) -- (n2); 
	\draw[-, ultra thick, color=black!70] (n2) -- (n3); 
    \draw[-, ultra thick, color=black!70] (n1) -- (n3); 
    \draw[-, ultra thick, dashed, color=black!70] (n2) -- (n5);
    \draw[-, ultra thick, dashed, color=black!70] (n3) -- (n4);
    \draw[-, ultra thick, dashed, color=black!70] (n4) -- (n5); 
    \draw[-, ultra thick, color=black!70] (n1) -- (n2); 
    \draw[-, ultra thick, color=black!70] (n4) -- (n6); 
    \draw[-, ultra thick, color=black!70] (n5) -- (n7); 
\end{tikzpicture}}
\par\end{centering}
\textcolor{black}{\caption{The integral graph of $\mathcal{G}(t)$ over time span $[t_{6l},t_{6(l+1)})$
where $l\in\mathbb{N}$.}
\label{fig:integral-network}}
\end{figure}
\textcolor{black}{The integral graph of $\mathcal{G}_{1}$, $\mathcal{G}_{2}$
and $\mathcal{G}_{3}$ over time span $[t_{6l},t_{6(l+1)})$, where
$l\in\mathbb{N}$, denoted by $\widetilde{\mathcal{G}}$, is shown
in Figure \ref{fig:integral-network}. One can see that it satisfies
the conditions in }\textbf{\textcolor{black}{Theorem }}\textcolor{black}{\ref{thm:cluster theorem}}\textbf{\textcolor{black}{,}}\textcolor{black}{{}
and the multi-agent system \eqref{equ:matrix-consensus-overall} on
the switching networks admits cluster consensus as shown in Figure
\ref{fig:trajectory}, which is the same as the system on the integral
network; see Figure \ref{fig:trajectory-2}. The cluster conditions
associated with the integral network of $\mathcal{G}(t)$ over one
period can therefore be employed to construct that applicable to switching
networks $\mathcal{G}(t)$ over $[0,\infty)$.}

\textcolor{black}{Consider a variant of the above Example by only
changing the matrix weights on edges $(3,4)$, $(2,5)$ and $(4,5)$
into 
\[
A_{34}(\mathcal{G}_{3})=-\left[\begin{array}{ccc}
4 & 0 & -2\\
0 & 2 & 1\\
-2 & 1 & 2
\end{array}\right],
\]
}

\textcolor{black}{
\[
A_{25}(\mathcal{G}_{3})=-\left[\begin{array}{ccc}
4 & 2 & 0\\
2 & 2 & 1\\
0 & 1 & 1
\end{array}\right],
\]
and 
\[
A_{45}(\mathcal{G}_{3})=\left[\begin{array}{ccc}
4 & 2 & 0\\
2 & 5 & 3\\
0 & 3 & 3
\end{array}\right],
\]
respectively. One can see that $\mathcal{G}(t)$ is simultaneously
structurally balanced and the integral graph of ${\color{black}\mathcal{G}(t)}$
over time span $[t_{6l},t_{6(l+1)})$, where $l\in\mathbb{N}$, denoted
by $\widetilde{\mathcal{G}}$, has a positive-negative spanning tree
$\mathcal{T}(\widetilde{\mathcal{G}})$. Therefore, according to }\textbf{\textcolor{black}{Corollary}}\textcolor{black}{{}
\ref{thm:structure-condition-bipartite-consensus}, the multi-agent
system \eqref{equ:matrix-consensus-overall} admits bipartite consensus;
see Figure \ref{fig:trajectory-3}.}
\begin{figure}
\begin{centering}
\textcolor{black}{\includegraphics[width=9cm]{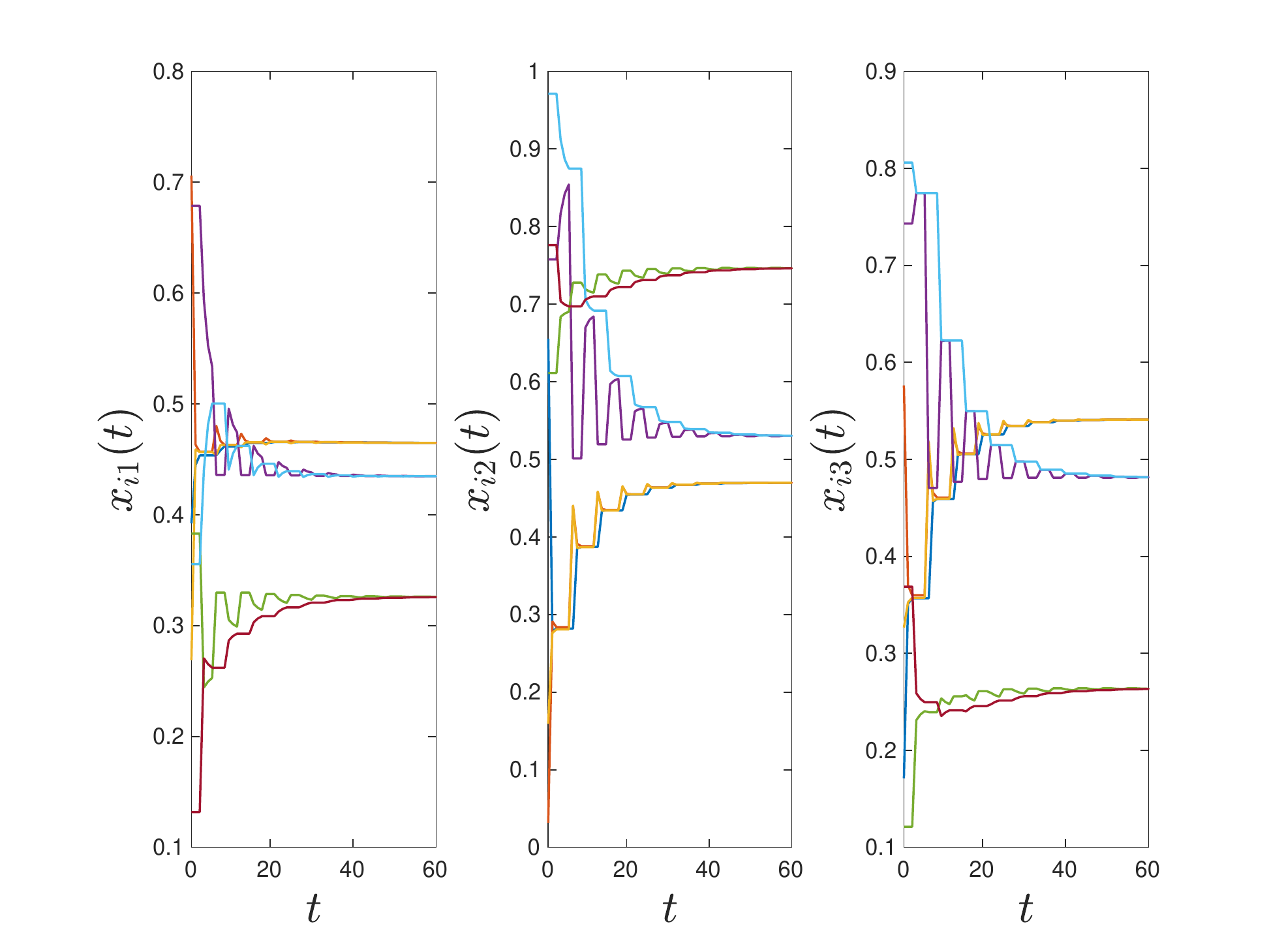}}
\par\end{centering}
\textcolor{black}{\caption{State evolution of the multi-agent system \eqref{equ:matrix-consensus-overall}
on a sequence of networks in Figure \ref{fig:example-network-1} with
switching sequences as \eqref{eq:sp-2}.}
\label{fig:trajectory}}
\end{figure}
\textcolor{black}{}
\begin{figure}
\begin{centering}
\textcolor{black}{\includegraphics[width=9cm]{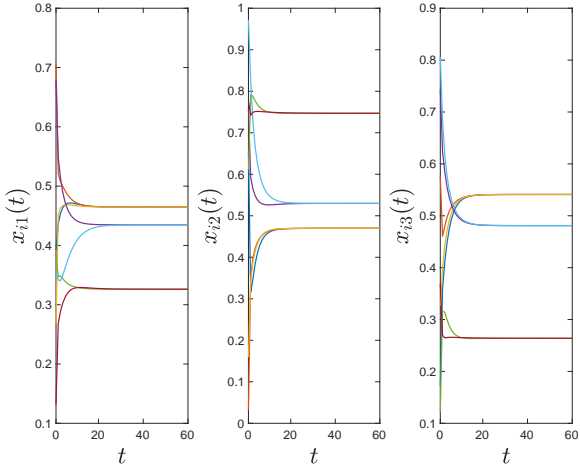}}
\par\end{centering}
\textcolor{black}{\caption{State evolution of the multi-agent system \eqref{equ:matrix-consensus-overall}
on the integral network in Figure \ref{fig:integral-network}.}
\label{fig:trajectory-2}}
\end{figure}
\textcolor{black}{}
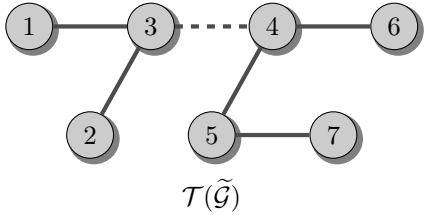
\begin{figure}
\begin{centering}
\textcolor{black}{\begin{tikzpicture}[scale=0.8]
	\node (n1) at (-3,0.3) [circle,circular drop shadow,fill=black!20,draw] {1};
    \node (n2) at (-2,-1.5) [circle,circular drop shadow,fill=black!20,draw] {2};
    \node (n3) at (-1 ,0.3) [circle,circular drop shadow,fill=black!20,draw] {3};
    \node (n4) at (1 ,0.3) [circle,circular drop shadow,fill=black!20,draw] {4};
	\node (n5) at (0,-1.5) [circle,circular drop shadow,fill=black!20,draw] {5};
    \node (n6) at (3 ,0.3) [circle,circular drop shadow,fill=black!20,draw] {6};
    \node (n7) at (2,-1.5) [circle,circular drop shadow,fill=black!20,draw] {7};  
    \node (G2) at (0,-2.5) {{$\mathcal{T}(\widetilde{\mathcal{G}}$})};
	\draw[-, ultra thick, color=black!70] (n2) -- (n3); 
    \draw[-, ultra thick, color=black!70] (n1) -- (n3); 
    \draw[-, ultra thick, dashed, color=black!70] (n3) -- (n4);
    \draw[-, ultra thick, color=black!70] (n4) -- (n5); 
    \draw[-, ultra thick, color=black!70] (n4) -- (n6); 
    \draw[-, ultra thick, color=black!70] (n5) -- (n7); 
\end{tikzpicture}}
\par\end{centering}
\textcolor{black}{\caption{The positive-negative spanning tree $\mathcal{T}(\widetilde{\mathcal{G}})$
of the integral graph $\mathcal{G}(t)$ over time span $[t_{6l},t_{6(l+1)})$
where $l\in\mathbb{N}$. \textcolor{black}{Those edges weighted by
positive definite matrices are illustrated by solid lines and edges
weighted by negative definite matrices are illustrated by dotted lines.}}
\label{fig:integral-network-1}}
\end{figure}
\textcolor{black}{}
\begin{figure}
\begin{centering}
\textcolor{black}{\includegraphics[width=9cm]{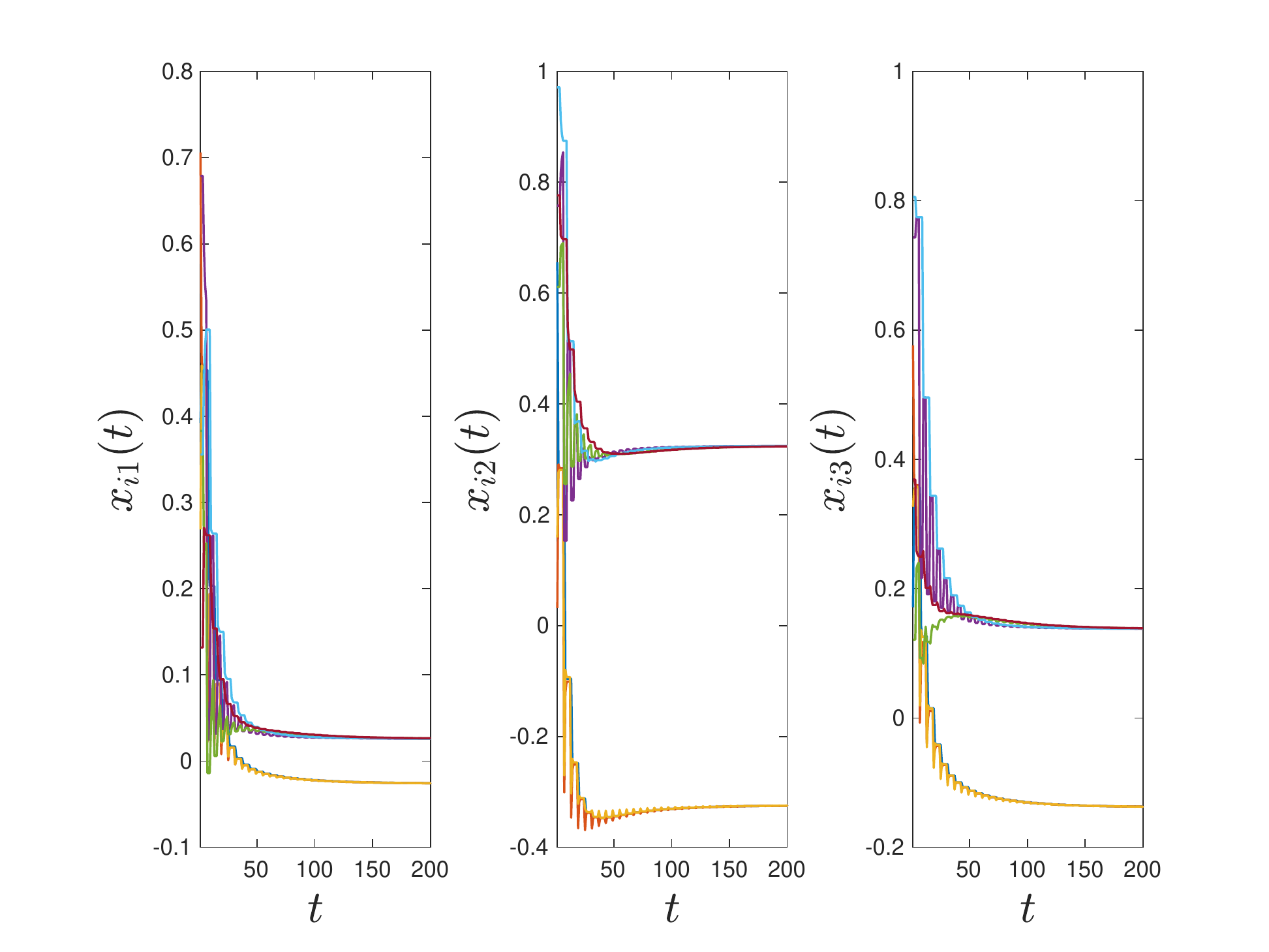}}
\par\end{centering}
\textcolor{black}{\caption{State evolution of the multi-agent system \eqref{equ:matrix-consensus-overall}
on a sequence of networks in Figure \ref{fig:example-network-1} with
switching sequences as \eqref{eq:sp-2}.}
\label{fig:trajectory-3}}
\end{figure}

\section{\textcolor{black}{Conclusion \label{sec:Conclusion}}}

This paper examines cluster consensus problems on matrix-weighted
switching networks. For such networks, necessary and/or sufficient
conditions for reaching cluster consensus that can be quantitatively
characterized are provided. It is shown that if the matrix-weighted
switching networks achieve the cluster consensus, then the cluster
consensus value belongs to the intersection of the null space of all
matrix-valued Laplacians. Furthermore, sufficient conditions for cluster
consensus are obtained using the matrix-valued Laplacian of the associated
integral network. In particular, conditions for bipartite consensus
is further provided under the condition the matrix-weighted switching
networks is simultaneously structurally balanced, as well as the explicit
expression of convergence state.

\section{\textcolor{black}{Appendix }}
\begin{lem}
\textcolor{black}{\label{lem:Rayleigh Theorem}\citet[p.235]{horn2012matrix}
(Rayleigh Theorem) Let $M\in\mathbb{R}^{n\times n}$ be symmetric
with eigenvalues $\lambda_{1}\leq\cdots\leq\lambda_{n}$. Let $\boldsymbol{x}_{1},\cdots,\boldsymbol{x}_{n}$
be corresponding mutually orthonormal vectors such that $M\boldsymbol{x}_{p}=\lambda_{p}\boldsymbol{x}_{p}$,
where $p\in\underline{n}$. Then, 
\[
\lambda_{1}\le\boldsymbol{x}^{\top}M\boldsymbol{x}\le\lambda_{n}
\]
for any unit vector $\boldsymbol{x}\in\mathbb{\mathbb{R}}^{n}$, with
equality in the right-hand (respectively, left-hand) inequality if
and only if $M\boldsymbol{x}\text{=}\lambda_{n}\boldsymbol{x}$ (respectively,
$M\boldsymbol{x}\text{=}\lambda_{1}\boldsymbol{x}$); moreover, 
\[
\lambda_{n}=\underset{\boldsymbol{x}\not={\bf 0}}{\text{{\bf max}}}\frac{\boldsymbol{x}^{\top}M\boldsymbol{x}}{\boldsymbol{x}^{\top}\boldsymbol{x}},
\]
and
\[
\lambda_{1}=\underset{\boldsymbol{x}\not={\bf 0}}{\text{{\bf min}}}\frac{\boldsymbol{x}^{\top}M\boldsymbol{x}}{\boldsymbol{x}^{\top}\boldsymbol{x}}.
\]
}
\end{lem}
\textcolor{black}{}
\textcolor{black}{We also make the observation that when $M\in\mathbb{R}^{n\times n}$
is positive semi-definite matrix, $\boldsymbol{x}^{\top}M\boldsymbol{x}=0$
if and only if $M\boldsymbol{x}=\boldsymbol{0}$.}
\begin{lem}
\textcolor{black}{\label{lem:convergence lemma}\citet{su2011stability}
Let $\{t_{k}|k\in\mathbb{N}\}$ be a sequence such that $t_{k+1}-t_{k}\geq\alpha>0$
for all $k\in\mathbb{N}$ and $t_{0}=0$. Suppose $F(t)$: $[0,\infty)\rightarrow\mathbb{R}$
satisfies }

\textcolor{black}{1) $\underset{t\rightarrow\infty}{\text{{\bf lim}}}F(t)$
exists;}

\textcolor{black}{2) $F(t)$ is twice differentiable on each interval
$[t_{k},t_{k+1})$;}

\textcolor{black}{3) $\ddot{F}(t)$ is bounded for $t\geq0$.}
\end{lem}
\textcolor{black}{Then $\underset{t\rightarrow\infty}{\text{{\bf lim}}}\dot{F}(t)=0$.}

\bibliographystyle{elsarticle-harv}
\addcontentsline{toc}{section}{\refname}\bibliography{mybib}

\end{document}